\documentclass{IEEEtran}
\IEEEoverridecommandlockouts

\makeatletter
\def\endthebibliography{%
  \def\@noitemerr{\@latex@warning{Empty `thebibliography' environment}}%
  \endlist
}
\makeatother

\usepackage[left=1.2cm, right=1.2cm, top=1.8cm, bottom=3.7cm]{geometry}
\usepackage[cmex10]{amsmath}
\usepackage[utf8]{inputenc} % allow utf-8 input
\usepackage[T1]{fontenc}    % use 8-bit T1 fonts
\usepackage{hyperref}       % hyperlinks
\usepackage{url}            % simple URL typesetting
\usepackage{booktabs}       % professional-quality tables
\usepackage{amsfonts}       % blackboard math symbols
\usepackage{nicefrac}       % compact symbols for 1/2, etc.
\usepackage{microtype}      % microtypography
\usepackage{xcolor}         % colors

\usepackage{amsmath,amssymb,amsthm}
\usepackage{graphicx}
\usepackage{comment}
\usepackage{flushend}

\newtheorem{lemma}{Lemma}
\newtheorem{theorem}{Theorem}

\newcommand*\diff{\mathop{}\!\mathrm{d}}

\title{Discriminative Mutual Information Estimators for Channel Capacity Learning}

\author{\IEEEauthorblockN{Nunzio A. Letizia \thanks{The authors are with the University of Klagenfurt - Chair of Embedded Communication Systems, 9020 Klagenfurt, Austria. (e-mail: \{nunzio.letizia, andrea.tonello\}@aau.at)} and Andrea M. Tonello} 
}

\begin{document}

\maketitle

\begin{abstract}
Channel capacity plays a crucial role in the development of modern communication systems as it represents the maximum rate at which information can be reliably transmitted over a communication channel. Nevertheless, for the majority of channels, finding a closed-form capacity expression remains an open challenge. This is because it requires to carry out two formidable tasks \textit{a)} the  computation of the mutual information between the channel input and output, and \textit{b)} its maximization with respect to the signal distribution at the channel input. 

In this paper, we address both tasks. Inspired by implicit generative models, we propose a novel cooperative framework to automatically learn the channel capacity, for any type of memory-less channel. In particular, we firstly develop a new methodology to estimate the mutual information directly from a discriminator typically deployed to train adversarial networks, referred to as \underline{di}scriminative \underline{m}utual information \underline{e}stimator (DIME). Secondly, we include the discriminator in a \underline{co}ope\underline{r}a\underline{ti}ve channel \underline{ca}pacity \underline{l}earning framework, referred to as CORTICAL, where a discriminator learns to distinguish between dependent and independent channel input-output samples while a generator learns to produce the optimal channel input distribution for which the discriminator exhibits the best performance. Lastly, we prove that a particular choice of the cooperative value function solves the channel capacity estimation problem. Simulation results demonstrate that the proposed method offers high accuracy.   
\end{abstract}

% Section: Introduction
\section{Introduction}
\label{sec:introduction}
Reliable transmission in a communication medium was firstly investigated in the milestone work of C. Shannon \cite{Shannon1948}. He suggested to partition the full communication system into a chain of three fundamental blocks, namely, the transmitter, the channel and the receiver. When possible, mathematics helps to model each block resulting in a full tractable system. Recently, machine learning techniques have subverted such bottom-up approach and also in communications, several data-driven models for physical layer \cite{Oshea2017,Karanov2018,Kim2018,Nachmani2018,Dorner2018,Raj2018} have been developed to overcome the lack of a channel model and optimal coding-decoding schemes.

In quantitative terms, the trade-off between the rate of transmission and reliability is expressed by the channel capacity. For a memory-less vector channel, the capacity is defined as
\begin{equation}
\label{eq:capacity}
C = \max_{p_X(\mathbf{x})} I(X;Y),
\end{equation}
where $p_X(\mathbf{x})$ is the input signal probability density function (PDF), $X$ is the channel input signal and $Y$ is the channel output signal and $I(X;Y)$ is the mutual information between $X$ and $Y$. 
However, estimating the channel capacity is a challenging task and few special cases, such as the additive white Gaussian noise (AWGN) channel case, have been solved so far \cite{Forney1999}. 
When the channel is not AWGN, the analytic study becomes mostly intractable leading to numerical solutions, relaxations or lower and upper bounds \cite{Arnold2006}. 
We argue that a data-driven approach can be pursued to learn the channel capacity. In the following we develop a cooperative framework aiming at estimating the capacity of a generic memory-less vector channel. In particular, we firstly propose two discriminative mutual information estimators, indirect and direct, based on adversarial training objectives. Secondly, we exploit the direct estimator to solve the channel capacity estimation problem, seen as a cooperative game between a generator and a discriminator. Lastly, we establish a connection to other variational lower bounds.

% Section: MINE and Capacity
\section{Related Work}
\label{sec:related}
The reinterpretation of the whole communication chain as an autoencoder-based system \cite{Oshea2017} pioneered a number of related works \cite{Dorner2018,TurboAE,Alberge2019} aimed at showing the potentiality of deep learning techniques applied to wireless communications. In this context, the autoencoder is represented by a deep neural network which takes as input a sequence of bits $\mathbf{s}$, produces a coded signal $\mathbf{x}$ fed into a channel layer with conditional probability density function $p_{Y}(\mathbf{y}|\mathbf{x})$, and tries to reconstruct the input sequence of bits from the channel output samples $\mathbf{y}$. Autoencoders are usually trained via cross-entropy minimization in order to reduce the block-error rate (BLER) of the full system. Only recently however, different metrics, such as bit-error rate (BER) metrics \cite{Cammerer2020} or the information rate \cite{Hoydis2019}, have been considered in the problem formulation. In particular, since the channel capacity constitutes the upper bound on the information rate, the maximization of the mutual information between the channel input and output has been discussed in \cite{Wunder2019,Letizia2021}. 

Traditional approaches for the mutual information estimation rely on binning, density and kernel estimation \cite{Moon1995} and k-nearest neighbours \cite{Kraskov2004}. Nevertheless, they do not scale to problems where high-dimensional data is present as it is the case in modern machine learning applications. Hence, deep neural networks have been recently leveraged to maximize variational lower bounds on the mutual information \cite{Nguyen2010, Mine2018, Poole2019a}. The expressive power of neural networks has shown promising results in this direction although less is known about the effectiveness of such estimators \cite{Song2020}, especially since they suffer from either high bias or high variance. 

Variational lower bounds estimators can be categorized in generative and discriminative approaches. The former attempts to separately estimate the joint and marginal distributions, $p_{XY}(\mathbf{x},\mathbf{y})$ and $p_{X}(\mathbf{x})\cdot p_{Y}(\mathbf{y})$, respectively. In \cite{Barber2003}, a variational distribution $q_{X,\phi}(\mathbf{x}|\mathbf{y})$ is introduced to replace the intractable conditional one $p_{X}(\mathbf{x}|\mathbf{y})$, leading to the bound
\begin{align}
\label{eq:Iba}
I(X;Y) \geq & I_{BA}(X;Y) \nonumber \\
= & \mathbb{E}_{(\mathbf{x},\mathbf{y})\sim p_{XY}(\mathbf{x},\mathbf{y})}\bigl[\log( q_{X,\phi}(\mathbf{x}|\mathbf{y}))
- \log( p_{X}(\mathbf{x})) \bigr] ,
\end{align}
where $q_{X,\phi}(\mathbf{x}|\mathbf{y})$ is a valid conditional distribution, parameterized by $\phi$. Instead, the discriminative approach attempts to directly estimate the density ratio (see \eqref{eq:density_ratio_1} in the following). It usually exploits an energy-based variational family of functions to provide a lower bound on the Kullback-Leibler (KL) divergence. As an example, the Donsker-Varadhan dual representation of the KL divergence \cite{Donsker1983} can be derived from \eqref{eq:Iba} \cite{Poole2019a} to produce the bound optimized in MINE \cite{Mine2018}
\begin{align}
\label{eq:MINE}
I(X;Y) \geq & I_{MINE}(X;Y) \nonumber \\
= & \sup_{\theta \in \Theta} \mathbb{E}_{(\mathbf{x},\mathbf{y})\sim p_{XY}(\mathbf{x},\mathbf{y})}[T_{\theta}(\mathbf{x},\mathbf{y})]  \nonumber \\
& - \log(\mathbb{E}_{(\mathbf{x},\mathbf{y})\sim p_X(\mathbf{x}) p_Y(\mathbf{y})}[e^{T_{\theta}(\mathbf{x},\mathbf{y})}]),
\end{align}
where $\theta \in \Theta$ parameterizes a family of functions $T_{\theta} : \mathcal{X}\times \mathcal{Y} \to \mathbb{R}$ through the use of a deep neural network. However, Monte Carlo sampling renders MINE a biased estimator. To avoid biased gradients, the authors in \cite{Mine2018} suggested to replace the partition function $\mathbb{E}_{p_X p_Y}[e^{T_{\theta}}]$ with an exponential moving average over mini-data-batches.

Another variational lower bound is based on the $f$-divergence representation introduced in \cite{Nguyen2010} (also referred to as $f$-MINE in \cite{Mine2018})
\begin{align}
\label{eq:NWJ}
I(X;Y) \geq & I_{NWJ}(X;Y) \nonumber \\
= & \sup_{\theta \in \Theta} \mathbb{E}_{(\mathbf{x},\mathbf{y})\sim p_{XY}(\mathbf{x},\mathbf{y})}[T_{\theta}(\mathbf{x},\mathbf{y})]  \nonumber \\
& -\mathbb{E}_{(\mathbf{x},\mathbf{y})\sim p_X(\mathbf{x}) p_Y(\mathbf{y})}[e^{T_{\theta}(\mathbf{x},\mathbf{y})-1}].
\end{align}
Although MINE provides a tighter bound $I_{MINE}\geq I_{NWJ}$, the NWJ estimator is unbiased. 

Both MINE and NWJ suffer from high-variance estimations and to combat such issue, the SMILE estimator was introduced in \cite{Song2020}. It is defined as
\begin{align}
\label{eq:SMILE}
I(X;Y)  \geq & I_{SMILE}(X;Y) \nonumber \\
= & \sup_{\theta \in \Theta} \mathbb{E}_{(\mathbf{x},\mathbf{y})\sim p_{XY}(\mathbf{x},\mathbf{y})}[T_{\theta}(\mathbf{x},\mathbf{y})] \nonumber \\
& -\log(\mathbb{E}_{(\mathbf{x},\mathbf{y})\sim p_X(\mathbf{x}) p_Y(\mathbf{y})}[\text{clip}(e^{T_{\theta}(\mathbf{x},\mathbf{y})},e^{-\tau},e^{\tau})]),
\end{align}
where $\text{clip}(v,l,u) = \max(\min(v,u),l)$ and it helps to obtain smoother partition functions estimates. SMILE is equivalent to MINE in the limit $\tau \to +\infty$. 

Another estimator referred to as noise-contrastive estimator (InfoNCE) \cite{NCE2018} is defined as 
\begin{align}
\label{eq:NCE}
I(X;Y) & \geq I_{NCE}(X;Y) \nonumber \\
& = \mathbb{E}_{(\mathbf{x},\mathbf{y})\sim p_{XY,N}(\mathbf{x},\mathbf{y})}\biggl[ \frac{1}{N} \sum_{i=1}^{N}{ \log\biggl( \frac{e^{T_{\theta}(\mathbf{x_i},\mathbf{y_i})}}{\frac{1}{N} \sum_{j=1}^{N}{e^{T_{\theta}(\mathbf{x_i},\mathbf{y_j})}}}\biggr)}  \biggr],
\end{align}
where $N$ is the batch size and $p_{XY,N}$ denotes the joint distribution of $N$ i.i.d. random variables sampled from $p_{XY}$. InfoNCE provides low variance estimates but it is upper bounded by $\log N$, resulting in a biased estimator. 

Inspired by the $f$-GAN training objectives \cite{Nowozin2016}, in the following we present two discriminative mutual information estimators that are based on the $f$-divergence measure but ultimately target the KL divergence.

% Section: Machine Learning
\subsection*{\textbf{Notation and Remarks}}
\label{subsec:notation}
$X$ denotes a multivariate random variable of dimension $d$, while $\mathbf{x}\in \mathcal{X}$ denotes its realization. $p_Y(\mathbf{y}|\mathbf{x})$ and $p_{XY}(\mathbf{x},\mathbf{y})$ represent the conditional and joint probability density or mass functions, while $p_X(\mathbf{x})p_Y(\mathbf{y})$ is the product of the two marginals probability functions. $I(X;Y)$ denotes the mutual information between the random variables $X$ and $Y$. Lastly, $D_{\text{KL}}(p||q)$ is the Kullback-Leibler divergence of the distribution $p$ from $q$.
All lemmas and theorems are proved in the Appendix.

\section{Discriminative Mutual Information Estimation}
\label{sec:MI}
The mutual information between two random variables, $X$ and $Y$, is a fundamental quantity in statistics, information theory and communication engineering as it plays a crucial role in the development of optimal communication schemes. It quantifies the statistical dependence (linear and non-linear \cite{Letizia2020}) between $X$ and $Y$ by measuring the amount of information obtained about $X$ via the observation of $Y$ at the channel output. It is defined as
\begin{equation}
\label{eq:mutual_information}
{I}(X;Y) = \mathbb{E}_{(\mathbf{x},\mathbf{y})\sim p_{XY}(\mathbf{x},\mathbf{y})}\biggl[\log\frac{p_{XY}(\mathbf{x},\mathbf{y})}{p_X(\mathbf{x})\cdot p_Y(\mathbf{y})}\biggr].
\end{equation}
\begin{comment}
\begin{equation}
\label{eq:mutual_information}
{I}(X;Y) = \mathbb{E}_{(\mathbf{x},\mathbf{y})\sim p_{XY}(\mathbf{x},\mathbf{y})}\biggl[\log\frac{p_{XY}(\mathbf{x},\mathbf{y})}{p_X(\mathbf{x})\cdot p_Y(\mathbf{y})}\biggr] = \mathbb{E}_{(\mathbf{x},\mathbf{y})\sim p_{XY}(\mathbf{x},\mathbf{y})}\bigl[\log p_{X}(\mathbf{x}|\mathbf{y}) - \log p_X(\mathbf{x})\bigr].
\end{equation}
\end{comment}

Unfortunately, computing $I(X;Y)$ is challenging since the joint probability density function $p_{XY}(\mathbf{x},\mathbf{y})$ and the marginals $p_X(\mathbf{x}),p_Y(\mathbf{y})$ are usually unknown, especially with high-dimensional data. Some recent techniques \cite{Papamakarios2017} have shown the ability to estimate probability density functions, exploiting neural networks that model the data dependence. However, to evaluate the mutual information, it is sufficient to compute the density ratio 
\begin{equation}
R(\mathbf{x},\mathbf{y}) = \frac{p_{XY}(\mathbf{x},\mathbf{y})}{p_X(\mathbf{x})\cdot p_Y(\mathbf{y})}
\label{eq:density_ratio_1}
\end{equation}
rather than the individual densities. In other words, the mutual information estimation can in principle enjoy the properties of implicit generative models, which are able to directly generate data that exhibits the same distribution of input data without any explicit density estimate. 
In this direction, the most successful technique is represented by generative adversarial networks (GANs), proposed in \cite{Goodfellow2014}. The main idea is to train a pair of networks in competition with each other: a generator network $G$ that captures the data distribution and a discriminator network $D$ that distinguishes if a sample is an original coming from real data rather than a fake coming from data generated by $G$. The training procedure for $G$ is to maximize the probability of $D$ making a mistake. GANs can be thought as a minimax two-player game which will end when a Nash equilibrium point is reached. Given an input noise vector $\mathbf{z}$ with PDF $p_{noise}(\mathbf{z})$, the mapping to the data space is achieved through $G(\mathbf{z};\theta_{gen})$. Defining the value function $V(G,D)$ as
\begin{align}
V(G,D) = & \; \mathbb{E}_{\mathbf{x} \sim p_{data}(\mathbf{x})}[\log D(\mathbf{x})] \nonumber \\ 
& + \mathbb{E}_{\mathbf{z} \sim p_{noise}(\mathbf{z})}[\log(1-D(G(\mathbf{z})))],
\label{eq:GAN}
\end{align}
it has been proved that the generator implicitly learns the true distribution. Indeed, $p_{gen} = p_{data}$ holds when the equilibrium is reached.

The adversarial training pushes the discriminator $D(\mathbf{x})$ towards the optimum value (see \cite{Goodfellow2014}) 
\begin{equation}
D^*(\mathbf{x}) = \frac{p_{data}(\mathbf{x})}{p_{data}(\mathbf{x})+p_{gen}(\mathbf{x})} = \frac{1}{1+\frac{p_{gen}(\mathbf{x})}{p_{data}(\mathbf{x})}},
\end{equation}
thus, given the value function in \eqref{eq:GAN}, the output of the optimum discriminator indirectly estimates the density ratio $p_{gen}/p_{data}$. 

\subsection{Indirect Discriminative Mutual Information Estimator}
The optimum discriminator (taken alone) in the GAN framework can be used as an indirect discriminative mutual information estimator (i-DIME). Since the calculation of the mutual information requires the density ratio \eqref{eq:density_ratio_1}, then $I(X;Y)$ can be indirectly estimated using the optimum discriminator $D^*$ when $p_{data}\equiv  p_{X}p_Y$ and $p_{gen} \equiv p_{XY}$. The following Lemma provides a mutual information estimator $I_{iDIME}(X;Y)$ that exploits the discriminator of the GAN value function. 

\begin{lemma}
\label{lemma:Lemma1}
Let $X\sim p_X(\mathbf{x})$ and $Y\sim p_{Y}(\mathbf{y}|\mathbf{x})$ be the channel input and output, respectively, such that $Y = H(X)$. Let $H(\cdot)$ be the stochastic channel model and $\pi(\cdot)$ be a permutation function such that  $p_{\pi(Y)}(\pi(\mathbf{y})|\mathbf{x}) = p_{Y}(\mathbf{y})$. 
If $\mathcal{J}_{}(D)$ is a value function defined as 
\begin{align}
\mathcal{J}_{}(D) = & \; \mathbb{E}_{\mathbf{x} \sim p_{X}(\mathbf{x})}\biggl[\log \biggl(D\biggl(\mathbf{x},\pi(H(\mathbf{x}))\biggr)\biggr)\biggr] \nonumber \\
& + \mathbb{E}_{\mathbf{x} \sim p_{X}(\mathbf{x})}\biggl[\log \biggl(1-D\biggl(\mathbf{x},H(\mathbf{x})\biggr)\biggr)\biggr],
\label{eq:discriminator_function}
\end{align}
then
\begin{equation}
\label{eq:optimal_discriminator_1}
D^*(\mathbf{x},\mathbf{y}) = \frac{p_{X}(\mathbf{x})\cdot p_Y(\mathbf{y})}{p_{XY}(\mathbf{x},\mathbf{y})+p_{X}(\mathbf{x})\cdot p_Y(\mathbf{y})} = \arg \max_D \mathcal{J}(D),
\end{equation}
and
\begin{align}
\label{eq:i-DMIE}
I(X;Y) & = I_{iDIME}(X;Y) \nonumber \\
& =  \mathbb{E}_{(\mathbf{x},\mathbf{y}) \sim p_{XY}(\mathbf{x},\mathbf{y})}\biggl[\log \biggl(\frac{1-D^*(\mathbf{x},\mathbf{y})}{D^*(\mathbf{x},\mathbf{y})} \biggr)\biggr].
\end{align}
\end{lemma}

Furthermore,  when $D(\mathbf{x},\mathbf{y})$ is parameterized by a neural network, the activation function of the last layer of a GAN discriminator is often chosen as the sigmoid function.  Therefore,  at the equilibrium
\begin{equation}
 D^*(\mathbf{x},\mathbf{y}; \theta) = \frac{1}{1+\exp(-\mathbf{W}^*_L \cdot \mathbf{z}_{L-1}-b^*_L)},
\end{equation}
where $\theta$ represents the parameters (weights and biases) of the neural network and $ \mathbf{z}_{L-1}$ is the output of the $(L-1)$-th layer.  Substituting the last layer output into the i-DIME estimator in \eqref{eq:i-DMIE} yields
\begin{equation}
I_{iDIME}(X;Y) =  -\mathbb{E}_{(\mathbf{x},\mathbf{y}) \sim p_{XY}(\mathbf{x},\mathbf{y})}\biggl[ \mathbf{W}^*_L \cdot \mathbf{z}_{L-1}+b^*_L  \biggr],
\end{equation}
thus, the second-to-last layer of a GAN discriminator provides the log-density ratio.

The following Lemma guarantees that the convergence of the indirect estimator $\hat{I}_{n,iDIME}(X;Y)$ is controlled by the convergence of $D$ towards the optimum GAN discriminator $D^*$ while optimizing $\mathcal{J}_{}(D)$.
\begin{lemma}
\label{lemma:Lemma2}
Let the discriminator $D(\cdot)$ be with enough capacity, i.e., in the non parametric limit. Consider the problem
\begin{equation}
D^* =  \; \arg \max_D \mathcal{J}_{}(D)
\label{eq:Lemma2_problem}
\end{equation}
where
\begin{align}
\mathcal{J}_{}(D) = & \; \mathbb{E}_{(\mathbf{x},\mathbf{y}) \sim p_{X}(\mathbf{x})\cdot p_Y(\mathbf{y})}\biggl[\log \biggl(D(\mathbf{x},\mathbf{y})\biggr)\biggr] \nonumber \\
& + \mathbb{E}_{(\mathbf{x},\mathbf{y}) \sim p_{XY}(\mathbf{x},\mathbf{y})}\biggl[\log \biggl(1-D(\mathbf{x},\mathbf{y})\biggr)\biggr],
\end{align}
and the update rule based on the gradient descent method
\begin{equation}
D^{(n+1)} = D^{(n)} + \mu \nabla \mathcal{J}_{}(D^{(n)}).
\end{equation}
If the gradient descent method converges to the global optimum $D^*$, the mutual information estimator 
\begin{equation}
\hat{I}_{n,iDIME}(X;Y) = \mathbb{E}_{(\mathbf{x},\mathbf{y}) \sim p_{XY}(\mathbf{x},\mathbf{y})}\biggl[\log \biggl(\frac{1-D^{(n)}(\mathbf{x},\mathbf{y})}{D^{(n)}(\mathbf{x},\mathbf{y})} \biggr)\biggr]
\end{equation}
converges to the real value of the mutual information $I(X;Y)$.
\end{lemma}

The maximization of the discriminator cost function in the GAN framework does not directly relate with a mutual information maximization. This is a consequence of the fact that the optimum GAN discriminator is not directly in the form needed for the mutual information estimation, namely, $p_{XY}/(p_X p_Y)$. To have an optimum discriminator in the desired density ratio form, in the next section we introduce the alpha-parameterized cost function $J_{\alpha}(D)$.

\subsection{Direct Discriminative Mutual Information Estimator}
The indirect GAN-based mutual information estimator (i-DIME) introduced in \eqref{eq:i-DMIE} has the advantage of exploiting the representation via the second-to-last layer,  nevertheless,  it suffers from two issues: 1)  The associated cost function does not directly maximize a lower bound on the mutual information.  2)  Even at the equilibrium,  if $p_{XY} >> p_X p_Y$ (i.e.  high values of the mutual information) or $p_{XY} << p_X p_Y$ (i.e.  low values of the mutual information),  $D^{(n)}(\mathbf{x},\mathbf{y})$ may saturate to either $0$ or $1$, providing unstable estimations.
To simultaneously tackle both issues,  we firstly design a cost function whose maximization coincides with a mutual information estimation, hence an optimum discriminator of the form $p_{XY}/(p_X p_Y)$.  Secondly,  we parameterize the new cost function via a positive constant $\alpha$ with the aim of controlling the estimation.  In particular,  $\alpha$ plays the role of a modulation factor that can be chosen to either amplify or attenuate vanishing or exploding density ratios,  respectively. 

The following Lemma introduces the alpha-parameterized cost function and sets the mathematical foundation behind the direct discriminative mutual information estimator (d-DIME).
\begin{lemma}
\label{lemma:Lemma3}
Let $X\sim p_X(\mathbf{x})$ and $Y\sim p_{Y}(\mathbf{y}|\mathbf{x})$ be the channel input and output, respectively, such that $Y = H(X)$. Let $H(\cdot)$ be the stochastic channel model and $\pi(\cdot)$ be a permutation function such that  $p_{\pi(Y)}(\pi(\mathbf{y})|\mathbf{x}) = p_{Y}(\mathbf{y})$.
If $\mathcal{J}_{\alpha}(D)$, $\alpha>0$, is a value function defined as 
\begin{align}
\mathcal{J}_{\alpha}(D) = \; & \alpha \cdot \mathbb{E}_{\mathbf{x} \sim p_{X}(\mathbf{x})}\biggl[\log \biggl(D\biggl(\mathbf{x},H(\mathbf{x})\biggr)\biggr)\biggr] \nonumber \\
& +\mathbb{E}_{\mathbf{x} \sim p_{X}(\mathbf{x})}\biggl[-D\biggl(\mathbf{x},\pi(H(\mathbf{x}))\biggr)\biggr],
\label{eq:discriminator_function}
\end{align}
then
\begin{equation}
\label{eq:optimal_discriminator_2}
D^*(\mathbf{x},\mathbf{y}) = \alpha \cdot \frac{p_{XY}(\mathbf{x},\mathbf{y})}{p_{X}(\mathbf{x})\cdot p_Y(\mathbf{y})} = \arg \max_D \mathcal{J}_{\alpha}(D).
\end{equation}
and
\begin{align}
I(X;Y) & = I_{dDIME}(X;Y) \nonumber \\
& = \mathbb{E}_{(\mathbf{x},\mathbf{y}) \sim p_{XY}(\mathbf{x},\mathbf{y})}\biggl[\log \biggl(\frac{D^*(\mathbf{x},\mathbf{y})}{\alpha} \biggr)\biggr]
\end{align}
or alternatively,
\begin{equation}
\label{eq:alpha_mutual_information}
I_{dDIME}(X;Y) = \frac{\mathcal{J}_{\alpha}(D^*)}{\alpha}+1-\log(\alpha).
\end{equation}
\end{lemma}

Conversely to the i-DIME estimator, this time the maximization of the discriminator cost function directly relates with the mutual information estimation as show in \eqref{eq:alpha_mutual_information}.  Therefore,  the following inequality holds
\begin{align}
I_{dDIME}(X;Y) & = \frac{\mathcal{J}_{\alpha}(D^*)}{\alpha}+1-\log(\alpha) \nonumber \\
& \geq \frac{\mathcal{J}_{\alpha}(D)}{\alpha}+1-\log(\alpha) = \tilde{I}_{dDIME}(X;Y),
\end{align}
concluding that for any type of positive discriminator function $D(\mathbf{x},\mathbf{y})$,  $\mathcal{J}_{\alpha}(D)/ \alpha + 1 - \log(\alpha)$ is a lower bound on the mutual information $I(X;Y)$.  Such lower bound will be better investigated in the context of variational lower bounds in Sec.\ref{sec:fenchel}.

The following Lemma guarantees that the convergence of the direct estimator $\hat{I}_{n,dDIME}(X;Y)$ is controlled by the convergence of $D$ towards the optimum discriminator $D^*$ while optimizing $\mathcal{J}_{\alpha}(D)$.

\begin{lemma}
\label{lemma:Lemma4}
Let the discriminator $D(\cdot)$ be with enough capacity, i.e., in the non parametric limit. Consider the problem
\begin{equation}
D^* =  \; \arg \max_D \mathcal{J}_{\alpha}(D)
\label{eq:Lemma3_problem}
\end{equation}
where
\begin{align}
\mathcal{J}_{\alpha}(D) = \; & \alpha \cdot \mathbb{E}_{(\mathbf{x},\mathbf{y}) \sim p_{XY}(\mathbf{x},\mathbf{y})}\biggl[\log \biggl(D(\mathbf{x},\mathbf{y})\biggr)\biggr] \nonumber \\
& +\mathbb{E}_{(\mathbf{x},\mathbf{y}) \sim p_{X}(\mathbf{x})\cdot p_Y(\mathbf{y})}\biggl[-D(\mathbf{x},\mathbf{y})\biggr],
\end{align}
and the update rule based on the gradient descent method
\begin{equation}
D^{(n+1)} = D^{(n)} + \mu \nabla \mathcal{J}_{\alpha}(D^{(n)}).
\end{equation}
If the gradient descent method converges to the global optimum $D^*$, the mutual information estimator 
\begin{equation}
\label{eq:hat_idMIE}
\hat{I}_{n,dDIME}(X;Y) = \mathbb{E}_{(\mathbf{x},\mathbf{y}) \sim p_{XY}(\mathbf{x},\mathbf{y})}\biggl[\log \biggl(\frac{D^{(n)}(\mathbf{x},\mathbf{y})}{\alpha} \biggr)\biggr]
\end{equation}
converges to the real value of the mutual information $I(X;Y)$.
\end{lemma}

Both i-DIME and d-DIME have been introduced as pure estimators given a certain input data. However, in communications, the transmitted samples/signals are the result of an encoding scheme that aims at approaching channel capacity and whose distribution has to be obtained through the maximization of the mutual information. In the following section, we propose a cooperative framework that combines an encoder/generator component with a d-DIME block to produce capacity-approaching codes. This provides a constructive method to estimate the channel capacity.

\section{CORTICAL: Cooperative Networks for Capacity Learning}
\label{sec:capacity}

\begin{figure*}[h]
	\centering
	\includegraphics[scale=0.43]{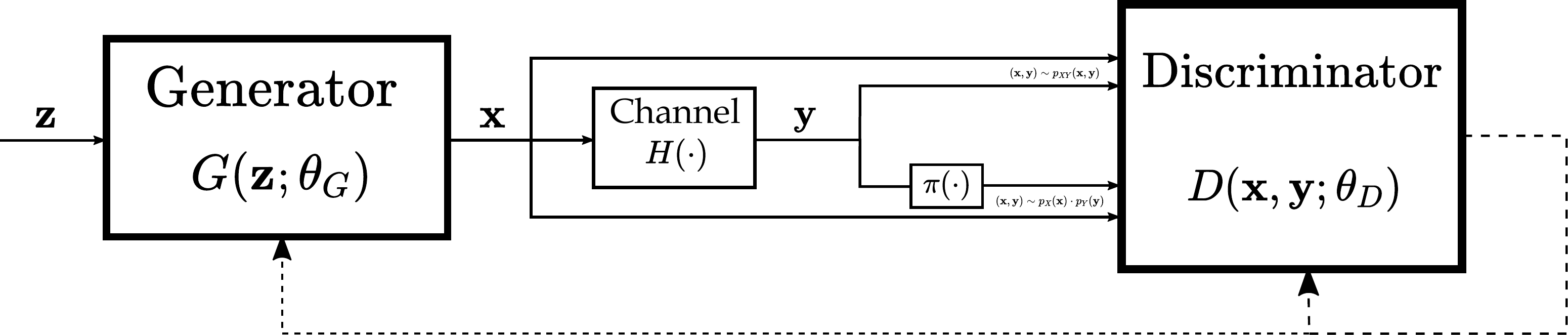}
	\caption{CORTICAL, Cooperative framework for capacity learning: a generator produces input samples with distribution $p_X(\mathbf{x})$ and a discriminator attempts to distinguish between paired and unpaired channel input-output samples.}
	\label{fig:Cooperative_networks}
\end{figure*} 

\begin{comment}
Recently, the communication chain has been reinterpreted as an autoencoder \cite{Oshea2017}. It consists of an encoder, a given channel model, and a decoder block which are jointly learned during the end-to-end training process. However, the autoencoder does not explicitly consider the achievable rate in the formulation and therefore may not produce capacity-approaching codes. 
In principle, since the decoding block acts as a classifier, it would be enough to separately design the optimum encoder, thus, the optimum code that maximizes the mutual information given the channel distribution, and learn the decoding scheme in a subsequent phase.
\end{comment}

In the following, we propose a direct DIME-based cooperative framework (CORTICAL) to learn the channel capacity (see \eqref{eq:capacity}). The approach is constructive in the sense that not only the capacity is estimated but also the optimum code that reaches it is built. We identify the encoder block as a generator whose objective is to produce channel input samples for which a discriminator exhibits the best performance in distinguishing (in the KL sense)  paired and unpaired channel input-output samples. Fig.\ref{fig:Cooperative_networks} illustrates the cooperative framework that learns both the channel input distribution $p_X(\mathbf{x})$ and the channel capacity $C$, given the channel model $H(\cdot)$.
The value function of the cooperative training scheme is discussed next.

\begin{theorem}
\label{theorem:Theorem1}
Let $X\sim p_X(\mathbf{x})$ and $Y\sim p_{Y}(\mathbf{y}|\mathbf{x})$ be the channel input and output, respectively, such that $X = G(Z)$ with $Z\sim p_Z(\mathbf{z})$. Let $Y = H(X)$ with $H(\cdot)$ being the stochastic channel model and let $\pi(\cdot)$ be a permutation function such that $p_{\pi(Y)}(\pi(\mathbf{y})|\mathbf{x}) = p_{Y}(\mathbf{y})$. 
If $\mathcal{J}_{\alpha}(G,D)$, $\alpha>0$, is a value function defined as 
\begin{align}
\mathcal{J}_{\alpha}(G,D) = \; & \alpha \cdot \mathbb{E}_{\mathbf{z} \sim p_{Z}(\mathbf{z})}\biggl[\log \biggl(D\biggl(G(\mathbf{z}),H(G(\mathbf{z}))\biggr)\biggr)\biggr] \nonumber \\
& +\mathbb{E}_{\mathbf{z} \sim p_{Z}(\mathbf{z})}\biggl[-D\biggl(G(\mathbf{z}),\pi(H(G(\mathbf{z})))\biggr)\biggr],
\label{eq:value_function}
\end{align}
then the channel capacity $C$ is the solution of
\begin{equation}
C = \max_{G} \max_{D} \frac{\mathcal{J}_{\alpha}(G,D)}{\alpha} + 1- \log(\alpha)
\end{equation}
\end{theorem}

The value function in \eqref{eq:value_function} can be interpreted as follows: given a certain source $\mathbf{z}$ with continuous or discrete (finite input alphabet of dimension $M$) probability density/mass function $p_Z(\mathbf{z})$, the generator maps the source into a code $\mathbf{x} = G(\mathbf{z})$ of density $p_X(\mathbf{x})$. Consequently, the generator and discriminator parameters adapt during the training process to produce paired channel input-output samples $(G(\mathbf{z}),H(G(\mathbf{z})))$ that maximize $\mathcal{J}_{\alpha}(G,D)$, thus, the mutual information $I(X;Y)$.

\section{Alpha Variational Representation of the Mutual Information}
\label{sec:fenchel}
In this section, we provide a variational interpretation of the estimation of the mutual information using the alpha-parameterized value function introduced in \eqref{eq:discriminator_function}. Inspired by the $f$-divergence variational estimation method proposed in \cite{Nguyen2010}, we exploit the Fenchel conjugate method to alternatively derive \eqref{eq:alpha_mutual_information}.
In particular, we show that \eqref{eq:alpha_mutual_information} constitutes a lower bound on the mutual information. Its maximization over positive discriminators $D$ (as proposed in Lemma 3) leads to the mutual information $I(X;Y)$. 
 
We firstly compute the Fenchel conjugate function $f^*(t)$ of $f(x)=-\alpha \cdot \log(u)$ with $\alpha>0$, and we use this result to recognize the alpha-parameterized value function in \eqref{eq:alpha_mutual_information} as a variational representation of the KL divergence. For the conjugate derivation please see Lemma \ref{lemma:Lemma5} in the Appendix.

\begin{comment}
\begin{lemma}
\label{lemma:Lemma5}
Let $f(u)=-\alpha \cdot \log(u)$, $\alpha>0$, be a convex, lower-semicontinuous function. Then, $f$ admits a convex conjugate function $f^*$, defined as
\begin{equation}
f^*(t) = \sup_{u\in \mathbb{R_+}} \{ut - f(u)\}
\end{equation}
with $t<0$ and expression
\begin{equation}
\label{eq:fenchel_conj}
f^*(t) = -\alpha - \alpha \cdot \log \biggl(-\frac{t}{\alpha}\biggr).
\end{equation}
\end{lemma}

\begin{proof}
From the Fenchel conjugate definition,
\begin{equation}
f^*(t) = \sup_{u\in \mathbb{R_+}} \{ut + \alpha \log(u)\},
\end{equation}
partial differentiation w.r.t. $u$ gives one stationary point in $u^*=-\alpha/t$ for negatives $t$. In particular, $u^*$ is a global maximum since 
\begin{equation}
\lim_{u\to 0^+} ut + \alpha \log(u)\ = \lim_{u\to +\infty} ut + \alpha \log(u)\ = -\infty.
\end{equation}
Hence, for $t\in \mathbb{R_{-}}$, 
\begin{equation}
f^*(t) = -\alpha +\alpha \cdot \log  \biggl(-\frac{\alpha}{t}\biggr).
\end{equation}
For $t\geq 0$, it is straightforward to verify that $f^*(t) = +\infty$.
\qedhere  
\end{proof}
\end{comment}

The next theorem introduces the alpha variational lower bound of the mutual information. The proof follows some ideas presented by Nguyen et al. \cite{Nguyen2010} for the general $f$-divergence variational representation. 

\begin{theorem}
\label{theorem:Theorem2}
For any class of functions $\mathcal{D}$ mapping from the sample domain $\mathcal{X}$ to $\mathbb{R_+}$, the following lower bound on the mutual information holds with $\alpha>0$
\begin{align}
I(X;Y) \geq & \sup_{D\in \mathcal{D_+}} \biggl\{\mathbb{E}_{(\mathbf{x},\mathbf{y}) \sim p_{XY}(\mathbf{x},\mathbf{y})}\bigl[\log (D(\mathbf{x},\mathbf{y}))\bigr] \nonumber \\
& - \frac{1}{\alpha}\cdot \mathbb{E}_{(\mathbf{x},\mathbf{y}) \sim p_{X}(\mathbf{x})\cdot p_{Y}(\mathbf{y})}\bigl[D(\mathbf{x},\mathbf{y})\bigr] +1 -\log(\alpha)\biggr\}
\end{align}
\end{theorem}

It is easy to recognize that the optimization problem proposed in Lemma \ref{lemma:Lemma3} exactly corresponds to a maximization of the mutual information alpha variational lower bound presented above. The choice of the value function in \eqref{eq:discriminator_function} was designed to output the density ratio $R$, required for the direct mutual information estimation. However, it is immediate to notice (see Lemma \ref{lemma:Lemma5}) that the Fenchel conjugate achieves the maximum in the desired density ratio point $u^*=-\alpha/t = \alpha \cdot p_{XY} / ( p_{X}\cdot p_{Y})$, matching the results of Lemma 3. 

Interestingly, the variational lower bound derived using the conjugate function approach corresponds to the scalar case of $I_{TUBA}$ in \cite{Poole2019a}, obtained from the inequality $\log(x)\leq \frac{x}{a}+\log(a)-1$ and with the substitution $f(\mathbf{x},\mathbf{y}) = \log(D(\mathbf{x},\mathbf{y}))$. Moreover, the maximization of the special case $\mathcal{J}_{1}$ leads to $I_{NWJ}$ \cite{Nguyen2010}.

% Section: Discussion
\section{Experimental Results}
\label{sec:results}

\begin{figure*}[h]
	\centering
	\includegraphics[scale=0.21]{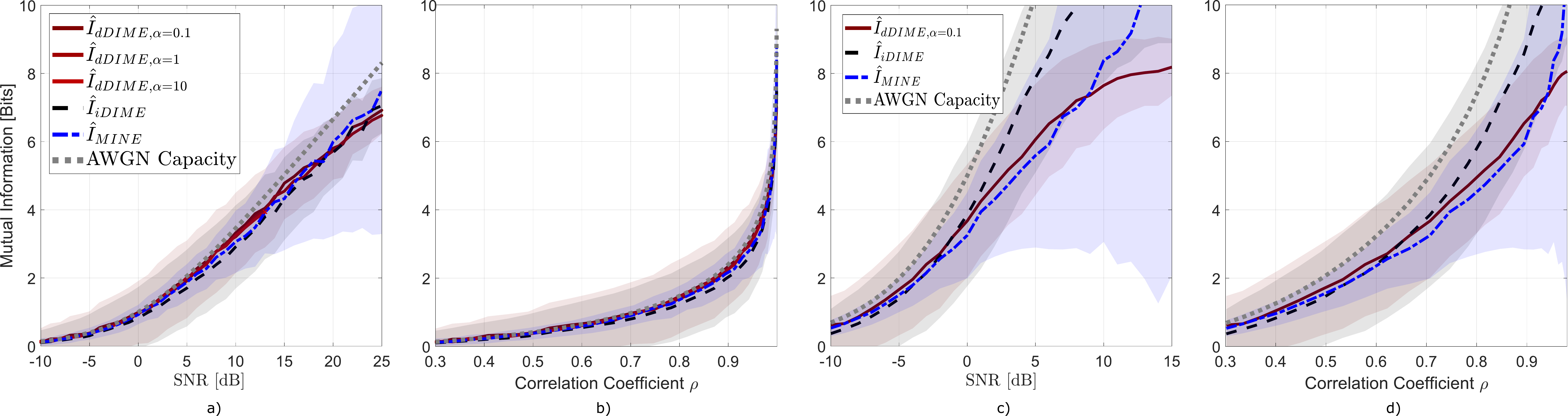}
	\caption{Performance of mutual information estimation approaches $\hat{I}$ between two multivariate Gaussians: a) SNR vs mutual information between a pair of $2$-d Gaussians; b) Correlation coefficient vs mutual information between a pair of $2$-d Gaussians; c) SNR vs mutual information between a pair of $10$-d Gaussians; d) Correlation coefficient vs mutual information between a pair of $10$-d Gaussians.}
	\label{fig:i_estimator}
\end{figure*}

\begin{figure*}[h]
	\centering
	\includegraphics[scale=0.21]{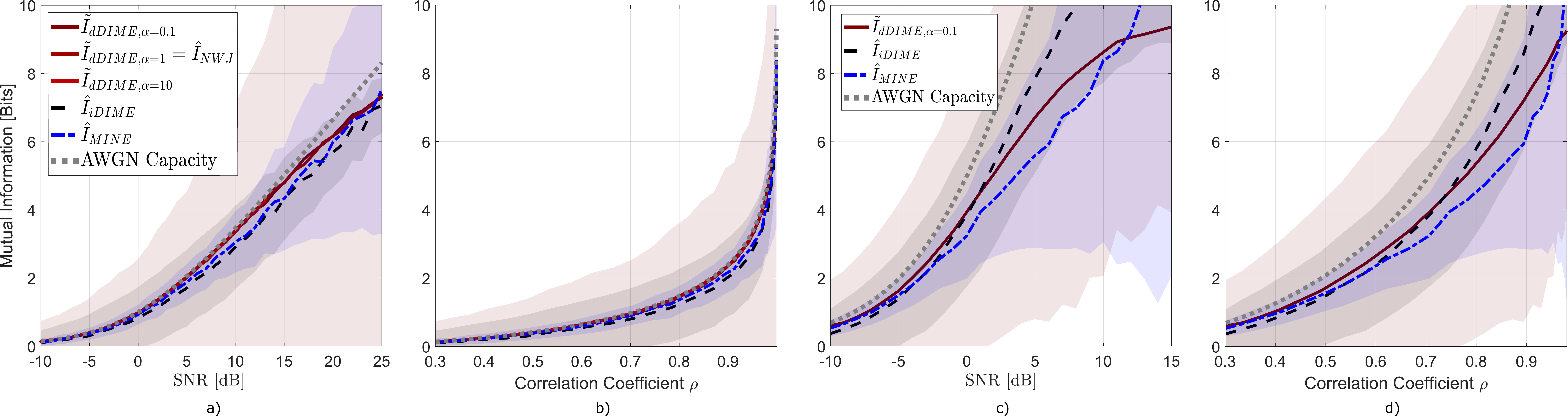}
	\caption{Performance of mutual information estimation approaches $\tilde{I}$ between two multivariate Gaussians: a) SNR vs mutual information between a pair of $2$-d Gaussians; b) Correlation coefficient vs mutual information between a pair of $2$-d Gaussians; c) SNR vs mutual information between a pair of $10$-d Gaussians; d) Correlation coefficient vs mutual information between a pair of $10$-d Gaussians.}
	\label{fig:j_estimator}
\end{figure*}

In the experimental results, we consider the transmission over an AWGN channel, for which we know the exact closed form expression of the mutual information when transmitting Gaussian samples, and for which we know the channel capacity under an average power constraint.

We firstly evaluate the accuracy of the proposed estimators. In particular,  we compare both i-DIME (denoted as $\hat{I}_{iDIME}$, see \eqref{eq:i-DMIE}) and d-DIME,  where the latter can assume two different forms. Either the estimator discussed in Lemma \ref{lemma:Lemma4} (see \eqref{eq:hat_idMIE}) or, exploiting the variational lower bound on the mutual information
\begin{equation}
\tilde{I}_{n,dDIME}(X;Y) := \frac{\mathcal{J}_{\alpha}(D^{(n)})}{\alpha}+1-\log(\alpha),
\end{equation}
where $D^{(n)}$ is the discriminator obtained after $n$ training iterations. 

Lastly, we study an analog and digital transmission scheme to show that the proposed cooperative networks are capable of learning the channel capacity without any prior assumption on the channel input distribution $p_X(\mathbf{x})$.

Both the generator $G$ and the discriminator $D$ are parameterized by neural networks. We train all discriminators models for $n = 5$k iterations. We use TensorFlow \cite{Tensorflow2016} and Keras \cite{Keras2015} to implement the proposed models. Details are discussed in the Appendix.

\subsection{Discriminative Estimators Performance}
\begin{figure*}[h]
	\centering
	\includegraphics[scale=0.175]{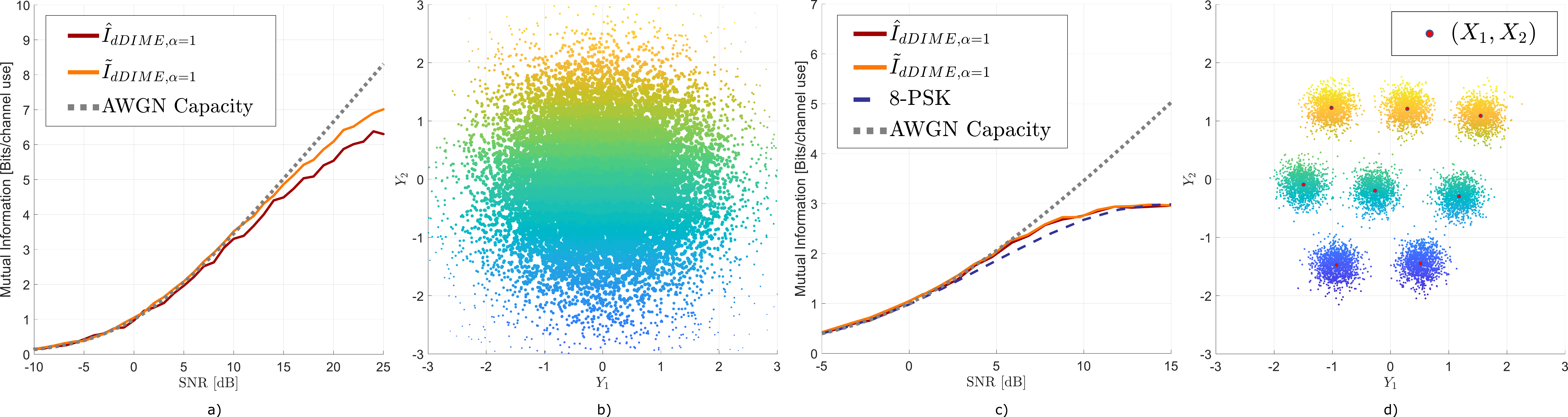}
	\caption{Performance of the CORTICAL framework in terms of maximal mutual information estimation using $\hat{I}_{dDIME}$ and $\tilde{I}_{dDIME}$ when : a) continuous $\mathbf{z}$ and $2$-dimensional $\mathbf{x}$; b) channel output samples $\mathbf{y}$ with continuous generator input $\mathbf{z}$ and $\text{SNR}=15$ dB c) discrete $\mathbf{z}$ ($M=8$) and $2$-dimensional $\mathbf{x}$; d) channel output samples $\mathbf{y}$ with discrete generator input $\mathbf{z}$ and $\text{SNR}=10$ dB.}
	\label{fig:capacity_plot}
\end{figure*} 

It is known that for a discrete memory-less channel with input-output relation as
$Y_i = X_i + N_i$, where the noise complex-valued samples $N_i \sim \mathcal{CN}(0,\sigma^2)$ are i.i.d. and independent of $X_i$, and with a power constraint on the input signal $\mathbb{E}[|X_i|^2] \leq P$, the mutual information $I(X_i;Y_i)$ is maximized when $X_i \sim \mathcal{CN}(0,P)$ and it corresponds to the channel capacity $C = \log_2(1+\text{SNR})$, where SNR denotes the signal-to-noise power ratio. 
When dealing with a pair of input-output $d$-dimensional Gaussians $X$ and $Y$, each with i.i.d. zero-mean components, the capacity can be written in terms of correlation coefficient \cite{Kraskov2004} $\rho = \mathbb{E}[X_i \cdot Y_i]/ \sqrt{\mathbb{E}[X_i^2]\cdot \mathbb{E}[Y_i^2]}$, for $i=1,\dots,d$, as 
\begin{equation}
I(X;Y) = -\frac{d}{2} \log_2(1-\rho^2).
\end{equation}
The correlation coefficient and the noise power $\sigma^2$ are related by $\rho^2 = 1/(1+\sigma^2)$ when $P=1$. The discriminative estimator block $D^{(n)}(\mathbf{x},\mathbf{y})$ receives both paired and unpaired samples of the $d$-dimensional Gaussians channel input-output samples. The discriminator output is used to estimate the mutual information with one of the discussed estimators. Despite $\hat{I}_{n,dDIME}$ and $\tilde{I}_{n,dDIME}$ being the same at the global maximum, they possess different properties. 

Fig.\ref{fig:i_estimator} illustrates the performance of the mutual information estimators $\hat{I}_{iDIME}$, $\hat{I}_{MINE}$ and $\hat{I}_{dDIME}$ for three different values of the parameter $\alpha$. For both tasks with $2$-d Gaussian samples and $10$-d Gaussian samples, $\hat{I}_{dDIME}$ has slightly better accuracy than $\hat{I}_{MINE}$ and has much lower variance. The shadow colored regions show the variance of each estimator which has been evaluated given the optimal discriminator and given batches of size $512$ taken from the empirical joint and marginal distributions, similarly to the analysis conducted in Sec.3 of \cite{Poole2019a}. The dashed curves correspond to the final estimated mutual information. The parameter $\alpha$ seems to have an impact on the accuracy only for low SNRs, since at high SNRs, high values of $\alpha$ result in numerical issues.
$\hat{I}_{iDIME}$ performs better when dealing with higher dimensional data. 
Fig.\ref{fig:j_estimator}, instead, illustrates the performance of the mutual information estimators $\hat{I}_{iDIME}$, $\hat{I}_{MINE}$ and $\tilde{I}_{dDIME}$ for three different values of the parameter $\alpha$. For both tasks with $2$-d Gaussian samples and $10$-d Gaussian samples, $\tilde{I}_{dDIME}$ has a better accuracy than $\hat{I}_{MINE}$ but also than $\hat{I}_{dDIME}$ at the expense of a higher variance, similar to what happens with the variance of $\hat{I}_{NWJ}$, having an exponential behaviour as described in \cite{Song2020}. 

From Fig.\ref{fig:i_estimator} and \ref{fig:j_estimator} it is possible to make the following observations: $\hat{I}_{dDIME}$ has the advantage of using only samples from the joint distribution $p_{XY}$ during the estimation, and it has a low variance at the expense of a less precise estimation compared to $\tilde{I}_{dDIME}$. Viceversa, $\tilde{I}_{dDIME}$ uses both the joint and the product of marginals $p_{X}p_{Y}$ during the estimation, it is much accurate but it has higher variance compared to other methods.

\subsection{Cooperative Networks Performance}

We argued that the optimal channel input distribution $p_X(\mathbf{x})$ is modeled by the generator in the CORTICAL framework. In the following, we consider two cases, a continuous and a discrete input distribution $p_Z(\mathbf{z})$. In the former case, the generator maps $\mathbf{z}$ into a continuous $2$-dimensional  channel input distribution. The latter case, instead, maps $M=8$ possible messages into bidimensional constellation points having probability mass function $p_X(\mathbf{x})$.

Fig.\ref{fig:capacity_plot} illustrates the estimated maximal mutual information using $\hat{I}_{dDIME}$ and $\tilde{I}_{dDIME}$ for the continuous and discrete case and shows the constellation diagram (channel output samples). In particular, Fig.\ref{fig:capacity_plot}a shows the estimated channel capacity when transmitting over an AWGN channel. The generator automatically learns to map the continuous signal $\mathbf{z}$ into a $2$-d Gaussian distribution whose noisy samples (for a fixed SNR) are depicted in Fig.\ref{fig:capacity_plot}b. Similarly, Fig.\ref{fig:capacity_plot}c shows the estimated mutual information when the input $\mathbf{z}$ can assume only $8$ possible values and it is compared with the established $8$-PSK modulation scheme. The transmitted and received samples (for a fixed SNR value) are placed in the constellation of Fig.\ref{fig:capacity_plot}d.

% remember to mention From an implementation point of view,  the activation function of the last layer of a GAN discriminator is often chosen as the sigmoid function. 

% Section: Conclusions
\section{Conclusions}
\label{sec:conclusions}
In this paper, we proposed CORTICAL, a cooperative framework to learn the channel capacity for any type of memory-less vector channel. CORTICAL consists of a generator network, acting as an encoder block that builds the optimal channel input distribution, and of a discriminator network acting as a mutual information estimator. We showed how the discriminator of a GAN can be used to indirectly estimate the mutual information between paired channel input-output samples. We also derived an alpha-parameterized cost function which allows direct estimation and which constitutes a lower bound on the true mutual information. We formalized such estimator in the context of variational representations of the mutual information. We lastly discussed the accuracy of the proposed estimators, denoted with $\hat{I}_{iDIME}$, $\hat{I}_{dDIME}$ and $\tilde{I}_{iDIME}$ and we demonstrated the advantages and disadvantages of the proposed approaches compared to $\hat{I}_{MINE}$. Given the simple formulation and the initial promising performance, we believe that the discriminative mutual information estimator combined with a generator block, as in the CORTICAL framework, is a good candidate to build optimal channel coding schemes. However, further investigations are needed for the architecture design and to analyze the behaviour of the estimators for high values of SNRs and for longer codes.

\bibliographystyle{unsrt}
\bibliography{biblio}

\begin{thebibliography}{10}

\bibitem{Shannon1948}
Claude~Elwood Shannon.
\newblock A mathematical theory of communication.
\newblock {\em The Bell System Technical Journal}, 27(3):379--423, 7 1948.

\bibitem{Oshea2017}
Tim {O'Shea} and Jakob {Hoydis}.
\newblock An introduction to deep learning for the physical layer.
\newblock {\em IEEE Transactions on Cognitive Communications and Networking},
  3(4):563--575, Dec 2017.

\bibitem{Karanov2018}
Boris Karanov, Mathieu Chagnon, Félix Thouin, Tobias~A. Eriksson, Henning
  Bülow, Domaniç Lavery, Polina Bayvel, and Laurent Schmalen.
\newblock End-to-end deep learning of optical fiber communications.
\newblock {\em Journal of Lightwave Technology}, 36(20):4843--4855, 2018.

\bibitem{Kim2018}
Hyeji Kim, Yihan Jiang, Sreeram Kannan, Sewoong Oh, and Pramod Viswanath.
\newblock Deepcode: Feedback codes via deep learning.
\newblock In S.~Bengio, H.~Wallach, H.~Larochelle, K.~Grauman, N.~Cesa-Bianchi,
  and R.~Garnett, editors, {\em Advances in Neural Information Processing
  Systems}, volume~31. Curran Associates, Inc., 2018.

\bibitem{Nachmani2018}
Eliya Nachmani, Elad Marciano, Loren Lugosch, Warren~J. Gross, David Burshtein,
  and Yair Be’ery.
\newblock Deep learning methods for improved decoding of linear codes.
\newblock {\em IEEE Journal of Selected Topics in Signal Processing},
  12(1):119--131, 2018.

\bibitem{Dorner2018}
Sebastian {Dorner}, Sebastian {Cammerer}, Jakob {Hoydis}, and Stephan ten
  {Brink}.
\newblock Deep learning based communication over the air.
\newblock {\em IEEE Journal of Selected Topics in Signal Processing},
  12(1):132--143, Feb 2018.

\bibitem{Raj2018}
Vishnu Raj and Sheetal Kalyani.
\newblock Backpropagating through the air: Deep learning at physical layer
  without channel models.
\newblock {\em IEEE Communications Letters}, 22(11):2278--2281, 2018.

\bibitem{Forney1999}
G.David Forney and Gottfried Ungerboeck.
\newblock Modulation and coding for linear gaussian channels.
\newblock {\em IEEE Transactions on Information Theory}, 44(6):2384--2415,
  1998.

\bibitem{Arnold2006}
Dieter~M. Arnold, Hans-Andrea Loeliger, Pascal~O. Vontobel, Aleksandar Kavcic,
  and Wei Zeng.
\newblock Simulation-based computation of information rates for channels with
  memory.
\newblock {\em IEEE Transactions on Information Theory}, 52(8):3498--3508,
  2006.

\bibitem{TurboAE}
Yihan Jiang, Hyeji Kim, Himanshu Asnani, Sreeram Kannan, Sewoong Oh, and Pramod
  Viswanath.
\newblock Turbo autoencoder: Deep learning based channel codes for
  point-to-point communication channels.
\newblock In {\em Advances in Neural Information Processing Systems 32}, pages
  2758--2768. Curran Associates, Inc., 2019.

\bibitem{Alberge2019}
Florence {Alberge}.
\newblock Deep learning constellation design for the awgn channel with additive
  radar interference.
\newblock {\em IEEE Transactions on Communications}, 67(2):1413--1423, 2019.

\bibitem{Cammerer2020}
Sebastian {Cammerer}, Fayçal~Ait {Aoudia}, Sebastian {Dörner}, Maximilian
  {Stark}, Jakob {Hoydis}, and Stephan {ten Brink}.
\newblock Trainable communication systems: Concepts and prototype.
\newblock {\em IEEE Transactions on Communications}, 68(9):5489--5503, 2020.

\bibitem{Hoydis2019}
Maximilian {Stark}, Fayçal {Ait Aoudia}, and Jakob {Hoydis}.
\newblock Joint learning of geometric and probabilistic constellation shaping.
\newblock In {\em 2019 IEEE Globecom Workshops (GC Wkshps)}, pages 1--6, 2019.

\bibitem{Wunder2019}
Rick {Fritschek}, Rafael~F. {Schaefer}, and Gerhard {Wunder}.
\newblock Deep learning for channel coding via neural mutual information
  estimation.
\newblock In {\em 2019 IEEE 20th International Workshop on Signal Processing
  Advances in Wireless Communications (SPAWC)}, pages 1--5, 2019.

\bibitem{Letizia2021}
Nunzio~A. Letizia and Andrea~M. Tonello.
\newblock Capacity-driven autoencoders for communications.
\newblock {\em IEEE Open Journal of the Communications Society}, 2:1366--1378,
  2021.

\bibitem{Moon1995}
Young-Il Moon, Balaji Rajagopalan, and Upmanu Lall.
\newblock Estimation of mutual information using kernel density estimators.
\newblock {\em Phys. Rev. E}, 52:2318--2321, Sep 1995.

\bibitem{Kraskov2004}
Alexander Kraskov, Harald St\"ogbauer, and Peter Grassberger.
\newblock Estimating mutual information.
\newblock {\em Phys. Rev. E}, 69:066138, Jun 2004.

\bibitem{Nguyen2010}
XuanLong Nguyen, Martin~J. Wainwright, and Michael~I. Jordan.
\newblock Estimating divergence functionals and the likelihood ratio by convex
  risk minimization.
\newblock {\em IEEE Transactions on Information Theory}, 56(11):5847--5861,
  2010.

\bibitem{Mine2018}
Mohamed~Ishmael Belghazi, Aristide Baratin, Sai Rajeshwar, Sherjil Ozair,
  Yoshua Bengio, Aaron Courville, and Devon Hjelm.
\newblock Mutual information neural estimation.
\newblock In {\em Proceedings of the 35th International Conference on Machine
  Learning}, volume~80 of {\em Proceedings of Machine Learning Research}, pages
  531--540, Stockholmsmässan, Stockholm Sweden, 10--15 Jul 2018. PMLR.

\bibitem{Poole2019a}
Ben Poole, Sherjil Ozair, Aaron Van Den~Oord, Alex Alemi, and George Tucker.
\newblock On variational bounds of mutual information.
\newblock In {\em Proceedings of the 36th International Conference on Machine
  Learning}, volume~97 of {\em Proceedings of Machine Learning Research}, pages
  5171--5180. PMLR, 09--15 Jun 2019.

\bibitem{Song2020}
Jiaming Song and Stefano Ermon.
\newblock Understanding the limitations of variational mutual information
  estimators.
\newblock In {\em 8th International Conference on Learning Representations,
  {ICLR} 2020, Addis Ababa, Ethiopia, April 26-30, 2020}, 2020.

\bibitem{Barber2003}
David Barber and Felix~V. Agakov.
\newblock The im algorithm: A variational approach to information maximization.
\newblock In {\em NIPS}, pages 201--208, 2003.

\bibitem{Donsker1983}
M.~D. Donsker and S.~R.~S. Varadhan.
\newblock Asymptotic evaluation of certain markov process expectations for
  large time. iv.
\newblock {\em Communications on Pure and Applied Mathematics}, 36(2):183--212,
  1983.

\bibitem{NCE2018}
A{\"{a}}ron van~den Oord, Yazhe Li, and Oriol Vinyals.
\newblock Representation learning with contrastive predictive coding.
\newblock {\em CoRR}, abs/1807.03748, 2018.

\bibitem{Nowozin2016}
Sebastian Nowozin, Botond Cseke, and Ryota Tomioka.
\newblock f-gan: Training generative neural samplers using variational
  divergence minimization.
\newblock In D.~Lee, M.~Sugiyama, U.~Luxburg, I.~Guyon, and R.~Garnett,
  editors, {\em Advances in Neural Information Processing Systems}, volume~29.
  Curran Associates, Inc., 2016.

\bibitem{Letizia2020}
Nunzio~A. Letizia and Andrea~M. Tonello.
\newblock Segmented generative networks: Data generation in the uniform
  probability space.
\newblock {\em IEEE Transactions on Neural Networks and Learning Systems},
  pages 1--10, 2020.

\bibitem{Papamakarios2017}
George Papamakarios, Theo Pavlakou, and Iain Murray.
\newblock Masked autoregressive flow for density estimation.
\newblock In {\em Advances in Neural Information Processing Systems},
  volume~30. Curran Associates, Inc., 2017.

\bibitem{Goodfellow2014}
Ian Goodfellow, Jean Pouget-Abadie, Mehdi Mirza, Bing Xu, David Warde-Farley,
  Sherjil Ozair, Aaron Courville, and Yoshua Bengio.
\newblock Generative adversarial nets.
\newblock In {\em Advances in Neural Information Processing Systems},
  volume~27. Curran Associates, Inc., 2014.

\bibitem{Tensorflow2016}
Mart{\'{\i}}n Abadi et~al.
\newblock Tensorflow: Large-scale machine learning on heterogeneous distributed
  systems.
\newblock {\em CoRR}, abs/1603.04467, 2016.

\bibitem{Keras2015}
François Chollet.
\newblock keras.
\newblock \url{https://github.com/fchollet/keras}, 2015.

\end{thebibliography}

\newpage 
\appendix

\subsection{Omitted Proofs}
\label{sec:omitted_proofs}

 \setcounter{lemma}{0}
 \setcounter{theorem}{0}

\begin{lemma}
\label{lemma:Lemma1}
Let $X\sim p_X(\mathbf{x})$ and $Y\sim p_{Y}(\mathbf{y}|\mathbf{x})$ be the channel input and output, respectively, such that $Y = H(X)$. Let $H(\cdot)$ be the stochastic channel model and $\pi(\cdot)$ be a permutation function such that  $p_{\pi(Y)}(\pi(\mathbf{y})|\mathbf{x}) = p_{Y}(\mathbf{y})$. 
If $\mathcal{J}_{}(D)$ is a value function defined as 
\begin{align}
\mathcal{J}_{}(D) = \; & \mathbb{E}_{\mathbf{x} \sim p_{X}(\mathbf{x})}\biggl[\log \biggl(D\biggl(\mathbf{x},\pi(H(\mathbf{x}))\biggr)\biggr)\biggr] \nonumber \\
& + \mathbb{E}_{\mathbf{x} \sim p_{X}(\mathbf{x})}\biggl[\log \biggl(1-D\biggl(\mathbf{x},H(\mathbf{x})\biggr)\biggr)\biggr],
\end{align}
then
\begin{equation}
D^*(\mathbf{x},\mathbf{y}) = \frac{p_{X}(\mathbf{x})\cdot p_Y(\mathbf{y})}{p_{XY}(\mathbf{x},\mathbf{y})+p_{X}(\mathbf{x})\cdot p_Y(\mathbf{y})} = \arg \max_D \mathcal{J}(D),
\end{equation}
and
\begin{align}
I(X;Y) & = I_{iDIME}(X;Y) \nonumber \\
& =  \mathbb{E}_{(\mathbf{x},\mathbf{y}) \sim p_{XY}(\mathbf{x},\mathbf{y})}\biggl[\log \biggl(\frac{1-D^*(\mathbf{x},\mathbf{y})}{D^*(\mathbf{x},\mathbf{y})} \biggr)\biggr].
\end{align}
\end{lemma}

\begin{proof}
The proof follows some steps already introduced in \cite{Goodfellow2014}, with the main difference that the discriminator in the current set-up needs to distinguish between the joint $p_{XY}(\mathbf{x},\mathbf{y})$ and the product of marginals $p_{X}(\mathbf{x})\cdot p_Y(\mathbf{y})$ distributions.

From the hypothesis, the value function can be rewritten as
\begin{align}
\mathcal{J}_{}(D) = \; & \mathbb{E}_{(\mathbf{x},\mathbf{y}) \sim p_{X}(\mathbf{x})\cdot p_Y(\mathbf{y})}\biggl[\log \biggl(D(\mathbf{x},\mathbf{y})\biggr)\biggr] \nonumber \\
& + \mathbb{E}_{(\mathbf{x},\mathbf{y}) \sim p_{XY}(\mathbf{x},\mathbf{y})}\biggl[\log \biggl(1-D(\mathbf{x},\mathbf{y})\biggr)\biggr],
\label{eq:value_discriminator_p}
\end{align}
and using the Lebesgue integral to compute the expectation
\begin{align}
\mathcal{J}_{}(D) = & \int_{\mathbf{y}} \int_{\mathbf{x}} \biggl[p_{X}(\mathbf{x})\cdot p_Y(\mathbf{y}) \log \biggl(D(\mathbf{x},\mathbf{y})\biggr) \nonumber \\
& + p_{XY}(\mathbf{x},\mathbf{y}) \log \biggl(1-D(\mathbf{x},\mathbf{y})\biggr)\biggr] \diff \mathbf{x} \diff \mathbf{y}.
\end{align}
To maximize $\mathcal{J}_{}(D)$, a necessary and sufficient condition requires to take the derivative of the integrand with respect to $D$ and to set it to $0$, yielding the following equation in $D$
\begin{equation}
\frac{p_{X}(\mathbf{x})\cdot p_Y(\mathbf{y})}{D(\mathbf{x},\mathbf{y})} -\frac{p_{XY}(\mathbf{x},\mathbf{y})}{1-D(\mathbf{x},\mathbf{y})} =0,
\end{equation}
whose solution is the optimum discriminator
\begin{equation}
D^*(\mathbf{x},\mathbf{y}) = \frac{p_{X}(\mathbf{x})\cdot p_Y(\mathbf{y})}{p_{XY}(\mathbf{x},\mathbf{y})+p_{X}(\mathbf{x})\cdot p_Y(\mathbf{y})}.
\label{eq:optimal_discriminator_1}
\end{equation}
This is because $\mathcal{J}_{}(D^*)$ is a maximum since the second derivative $-\frac{p_{X}(\mathbf{x})\cdot p_Y(\mathbf{y})}{D^2(\mathbf{x},\mathbf{y})}-\frac{p_{XY}(\mathbf{x},\mathbf{y})}{(1-D(\mathbf{x},\mathbf{y}))^2}$ is a non-positive function.
Finally, at the equilibrium, 
\begin{equation}
\frac{p_{XY}(\mathbf{x},\mathbf{y})}{p_{X}(\mathbf{x})\cdot p_Y(\mathbf{y})} = \frac{1-D^*(\mathbf{x},\mathbf{y})}{D^*(\mathbf{x},\mathbf{y})},
\end{equation}
therefore the thesis follows, since the mutual information can be written as
\begin{align}
I(X;Y)  &= \mathbb{E}_{(\mathbf{x},\mathbf{y})\sim p_{XY}(\mathbf{x},\mathbf{y})}\biggl[\log\frac{p_{XY}(\mathbf{x},\mathbf{y})}{p_X(\mathbf{x})\cdot p_Y(\mathbf{y})}\biggr] \nonumber \\
&= I_{iDIME}(X;Y) \nonumber \\
& =  \mathbb{E}_{(\mathbf{x},\mathbf{y}) \sim p_{XY}(\mathbf{x},\mathbf{y})}\biggl[\log \biggl(\frac{1-D^*(\mathbf{x},\mathbf{y})}{D^*(\mathbf{x},\mathbf{y})} \biggr)\biggr].
\end{align}
\end{proof}

\begin{lemma}
\label{lemma:Lemma2}
Let the discriminator $D(\cdot)$ be with enough capacity, i.e., in the non parametric limit. Consider the problem
\begin{equation}
D^* =  \; \arg \max_D \mathcal{J}_{}(D)
\end{equation}
where
\begin{align}
\mathcal{J}_{}(D) = \; & \mathbb{E}_{(\mathbf{x},\mathbf{y}) \sim p_{X}(\mathbf{x})\cdot p_Y(\mathbf{y})}\biggl[\log \biggl(D(\mathbf{x},\mathbf{y})\biggr)\biggr] \nonumber \\
& + \mathbb{E}_{(\mathbf{x},\mathbf{y}) \sim p_{XY}(\mathbf{x},\mathbf{y})}\biggl[\log \biggl(1-D(\mathbf{x},\mathbf{y})\biggr)\biggr],
\end{align}
and the update rule based on the gradient descent method
\begin{equation}
D^{(n+1)} = D^{(n)} + \mu \nabla \mathcal{J}_{}(D^{(n)}).
\end{equation}
If the gradient descent method converges to the global optimum $D^*$, the mutual information estimator 
\begin{equation}
\hat{I}_{n,iDIME}(X;Y) = \mathbb{E}_{(\mathbf{x},\mathbf{y}) \sim p_{XY}(\mathbf{x},\mathbf{y})}\biggl[\log \biggl(\frac{1-D^{(n)}(\mathbf{x},\mathbf{y})}{D^{(n)}(\mathbf{x},\mathbf{y})} \biggr)\biggr]
\end{equation}
converges to the real value of the mutual information $I(X;Y)$.
\end{lemma}

\begin{proof}
The solution to \eqref{eq:Lemma2_problem} is given by \eqref{eq:optimal_discriminator_1} of Lemma \ref{lemma:Lemma1}. Let $\delta^{(n)}=D^*-D^{(n)}$ be the displacement between the optimum discriminator $D^*$ and the obtained one $D^{(n)}$ at the iteration $n$, then
\begin{align}
\hat{I}_{n,iDIME}(X;Y) & = \mathbb{E}_{(\mathbf{x},\mathbf{y}) \sim p_{XY}(\mathbf{x},\mathbf{y})}\biggl[\log \biggl(\frac{1-D^{(n)}(\mathbf{x},\mathbf{y})}{D^{(n)}(\mathbf{x},\mathbf{y})} \biggr)\biggr] \nonumber \\
& = \mathbb{E}_{(\mathbf{x},\mathbf{y}) \sim p_{XY}(\mathbf{x},\mathbf{y})}\biggl[\log \biggl(R^{(n)}(\mathbf{x},\mathbf{y}) \biggr)\biggr],
\end{align}
where $R^{(n)}(\mathbf{x},\mathbf{y})$ represents the estimated density ratio at the $n$-th iteration and it is related to the optimum ratio $R^*(\mathbf{x},\mathbf{y})$ as follows
\begin{align}
R^* - R^{(n)} & = \frac{1-D^*}{D^*} - \frac{1-D^{(n)}}{D^{(n)}} \nonumber \\  
& = \frac{1-D^*}{D^*} - \frac{1-(D^*-\delta^{(n)})}{D^*-\delta^{(n)}}  \nonumber \\
& = -\frac{\delta^{(n)}}{D^* (D^*-\delta^{(n)})}.
\end{align}
Therefore,
\begin{align}
& \hat{I}_{n,iDIME}(X;Y) = \mathbb{E}_{(\mathbf{x},\mathbf{y}) \sim p_{XY}(\mathbf{x},\mathbf{y})}\biggl[\log \bigl(R^{(n)}\bigr)\biggr] \nonumber \\
& =  \mathbb{E}_{(\mathbf{x},\mathbf{y}) \sim p_{XY}(\mathbf{x},\mathbf{y})}\biggl[\log \biggl(R^*+\frac{\delta^{(n)}}{D^* (D^*-\delta^{(n)})} \biggr)\biggr] \nonumber \\ 
& = \mathbb{E}_{(\mathbf{x},\mathbf{y}) \sim p_{XY}(\mathbf{x},\mathbf{y})}\biggl[\log \biggl(\bigl(R^*\bigr)\biggl(1+\frac{\delta^{(n)}}{R^*  D^* (D^*-\delta^{(n)})} \biggr)\biggr)\biggr] \nonumber \\ 
& = I(X;Y) + \mathbb{E}_{(\mathbf{x},\mathbf{y}) \sim p_{XY}(\mathbf{x},\mathbf{y})}\biggl[\log \biggl(1+\frac{\delta^{(n)}}{(1-D^*)(D^*-\delta^{(n)})} \biggr)\biggr].
\end{align}

If the gradient descent method converges towards the optimum solution $D^*$, $\delta^{(n)} \rightarrow 0$ and 
\begin{align}
& \hat{I}_{n,iDIME}(X;Y) \simeq  \nonumber \\
& \simeq \; I(X;Y) + \delta^{(n)} \cdot \mathbb{E}_{(\mathbf{x},\mathbf{y}) \sim p_{XY}(\mathbf{x},\mathbf{y})}\biggl[\frac{1}{(1-D^*)(D^*-\delta^{(n)})} \biggr] \nonumber \\ 
& \simeq \; I(X;Y) + \delta^{(n)} \int_{\mathbf{y}} \int_{\mathbf{x}}{\frac{(p_{XY}(\mathbf{x},\mathbf{y})+p_{X}(\mathbf{x})p_Y(\mathbf{y}))^2}{p_{XY}(\mathbf{x},\mathbf{y})} \diff \mathbf{x} \diff \mathbf{y}} \nonumber \\ 
& = \; I(X;Y) + 4\delta^{(n)} + \nonumber \\ 
& \;  + \delta^{(n)} \int_{\mathbf{y}} \int_{\mathbf{x}}{\frac{(p_{XY}(\mathbf{x},\mathbf{y})-p_{X}(\mathbf{x})p_Y(\mathbf{y}))^2}{p_{XY}(\mathbf{x},\mathbf{y})} \diff \mathbf{x} \diff \mathbf{y}} \nonumber \\ 
& = \; I(X;Y) + 4\delta^{(n)} + \delta^{(n)} \cdot \chi^2 \bigl(p_{X}p_{Y},p_{XY}\bigr),
\end{align}
where $\chi^2(Q,P)$ is the Chi-square divergence between the distributions $P$ and $Q$.
In the asymptotic limit ($n\rightarrow +\infty$), it holds that 
\begin{equation}
|I(X;Y)-\hat{I}_{n,iDIME}(X;Y)|\rightarrow 0.
\end{equation}
\end{proof}

\begin{lemma}
Let $X\sim p_X(\mathbf{x})$ and $Y\sim p_{Y}(\mathbf{y}|\mathbf{x})$ be the channel input and output, respectively, such that $Y = H(X)$. Let $H(\cdot)$ be the stochastic channel model and $\pi(\cdot)$ be a permutation function such that  $p_{\pi(Y)}(\pi(\mathbf{y})|\mathbf{x}) = p_{Y}(\mathbf{y})$. 
If $\mathcal{J}_{\alpha}(D)$, $\alpha>0$, is a value function defined as 
\begin{align}
\mathcal{J}_{\alpha}(D) = \; & \alpha \cdot \mathbb{E}_{\mathbf{x} \sim p_{X}(\mathbf{x})}\biggl[\log \biggl(D\biggl(\mathbf{x},H(\mathbf{x})\biggr)\biggr)\biggr] \nonumber \\
& +\mathbb{E}_{\mathbf{x} \sim p_{X}(\mathbf{x})}\biggl[-D\biggl(\mathbf{x},\pi(H(\mathbf{x}))\biggr)\biggr],
\end{align}
then
\begin{equation}
D^*(\mathbf{x},\mathbf{y}) = \alpha \cdot \frac{p_{XY}(\mathbf{x},\mathbf{y})}{p_{X}(\mathbf{x})\cdot p_Y(\mathbf{y})} = \arg \max_D \mathcal{J}_{\alpha}(D).
\end{equation}
and
\begin{align}
I(X;Y) & = I_{dDIME}(X;Y) \nonumber \\
& = \mathbb{E}_{(\mathbf{x},\mathbf{y}) \sim p_{XY}(\mathbf{x},\mathbf{y})}\biggl[\log \biggl(\frac{D^*(\mathbf{x},\mathbf{y})}{\alpha} \biggr)\biggr]
\end{align}
or alternatively,
\begin{equation}
I_{dDIME}(X;Y) = \frac{\mathcal{J}_{\alpha}(D^*)}{\alpha}+1-\log(\alpha).
\end{equation}
\end{lemma}

\begin{proof}
From the hypothesis, the value function can be rewritten as
\begin{align}
\mathcal{J}_{\alpha}(D) = \; & \alpha \cdot \mathbb{E}_{(\mathbf{x},\mathbf{y}) \sim p_{XY}(\mathbf{x},\mathbf{y})}\biggl[\log \biggl(D(\mathbf{x},\mathbf{y})\biggr)\biggr] \nonumber \\
& +\mathbb{E}_{(\mathbf{x},\mathbf{y}) \sim p_{X}(\mathbf{x})\cdot p_Y(\mathbf{y})}\biggl[-D(\mathbf{x},\mathbf{y})\biggr],
\label{eq:value_discriminator_p}
\end{align}
and using the Lebesgue integral to compute the expectation
\begin{align}
\mathcal{J}_{\alpha}(D) = \; & \alpha \cdot \int_{\mathbf{y}} \int_{\mathbf{x}}\biggl[p_{XY}(\mathbf{x},\mathbf{y}) \log \biggl(D(\mathbf{x},\mathbf{y})\biggr) \nonumber \\ 
& + p_{X}(\mathbf{x})\cdot p_Y(\mathbf{y}) \biggl(-D(\mathbf{x},\mathbf{y})\biggr)\biggr] \diff \mathbf{x} \diff \mathbf{y}.
\end{align}
Taking the derivative of integrand with respect to $D$ and setting it to $0$ yields the following equation in $D$
\begin{equation}
\alpha \cdot \frac{p_{XY}(\mathbf{x},\mathbf{y})}{D(\mathbf{x},\mathbf{y})} - p_{X}(\mathbf{x})\cdot p_Y(\mathbf{y})=0,
\end{equation}
whose solution is the optimum discriminator
\begin{equation}
D^*(\mathbf{x},\mathbf{y}) = \alpha \cdot \frac{p_{XY}(\mathbf{x},\mathbf{y})}{p_{X}(\mathbf{x})\cdot p_Y(\mathbf{y})}.
\end{equation}
In particular, $\mathcal{J}_{\alpha}(D^*)$ is a maximum since the second derivative $-\alpha \cdot \frac{p_{XY}(\mathbf{x},\mathbf{y})}{D^2(\mathbf{x},\mathbf{y})}$ is a non-positive function.
Finally, at the equilibrium, 
\begin{equation}
\frac{p_{XY}(\mathbf{x},\mathbf{y})}{p_{X}(\mathbf{x})\cdot p_Y(\mathbf{y})} = \frac{D^*(\mathbf{x},\mathbf{y})}{\alpha},
\end{equation}
therefore
\begin{align}
I(X;Y) = & \; I_{dDIME}(X;Y) \nonumber \\
 = & \;  \mathbb{E}_{(\mathbf{x},\mathbf{y}) \sim p_{XY}(\mathbf{x},\mathbf{y})}\biggl[\log \biggl(\frac{D^*(\mathbf{x},\mathbf{y})}{\alpha} \biggr)\biggr].
\end{align}
Furthermore, substituting $D^*(\mathbf{x},\mathbf{y})$ in \eqref{eq:value_discriminator_p} gives
\begin{align}
\mathcal{J}_{\alpha}(D^*) =  \; \alpha & \cdot \mathbb{E}_{(\mathbf{x},\mathbf{y}) \sim p_{XY}(\mathbf{x},\mathbf{y})}\biggl[\log \biggl(\alpha \cdot \frac{p_{XY}(\mathbf{x},\mathbf{y})}{p_{X}(\mathbf{x})\cdot p_Y(\mathbf{y})} \biggr)\biggr] \nonumber \\
& +\mathbb{E}_{(\mathbf{x},\mathbf{y}) \sim p_{X}(\mathbf{x})\cdot p_Y(\mathbf{y})}\biggl[-\alpha \cdot \frac{p_{XY}(\mathbf{x},\mathbf{y})}{p_{X}(\mathbf{x})\cdot p_Y(\mathbf{y})}\biggr],
\end{align}
where it is easy to recognize that the second term of the right hand side is equal to $-\alpha$, resulting in
\begin{align}
\mathcal{J}_{\alpha}(D^*) =  \; & \alpha \log(\alpha) - \alpha + \nonumber \\
& +\alpha \cdot \mathbb{E}_{(\mathbf{x},\mathbf{y}) \sim p_{XY}(\mathbf{x},\mathbf{y})}\biggl[\log \biggl(\frac{p_{XY}(\mathbf{x},\mathbf{y})}{p_{X}(\mathbf{x})\cdot p_Y(\mathbf{y})} \biggr)\biggr],
\end{align}
and the thesis follows
\begin{equation}
\mathcal{J}_{\alpha}(D^*) =  \; \alpha \log(\alpha) - \alpha +\alpha \cdot I_{dDIME}(X;Y).
\end{equation}
\end{proof}
\qedhere  

\begin{lemma}
Let the discriminator $D(\cdot)$ be with enough capacity, i.e., in the non parametric limit. Consider the problem
\begin{equation}
D^* =  \; \arg \max_D \mathcal{J}_{\alpha}(D)
\end{equation}
where
\begin{align}
\mathcal{J}_{\alpha}(D) = \; & \alpha \cdot \mathbb{E}_{(\mathbf{x},\mathbf{y}) \sim p_{XY}(\mathbf{x},\mathbf{y})}\biggl[\log \biggl(D(\mathbf{x},\mathbf{y})\biggr)\biggr] \nonumber \\
& +\mathbb{E}_{(\mathbf{x},\mathbf{y}) \sim p_{X}(\mathbf{x})\cdot p_Y(\mathbf{y})}\biggl[-D(\mathbf{x},\mathbf{y})\biggr],
\end{align}
and the update rule based on the gradient descent method
\begin{equation}
D^{(n+1)} = D^{(n)} + \mu \nabla \mathcal{J}_{\alpha}(D^{(n)}).
\end{equation}
If the gradient descent method converges to the global optimum $D^*$, the mutual information estimator 
\begin{equation}
\hat{I}_{n,dDIME}(X;Y) = \mathbb{E}_{(\mathbf{x},\mathbf{y}) \sim p_{XY}(\mathbf{x},\mathbf{y})}\biggl[\log \biggl(\frac{D^{(n)}(\mathbf{x},\mathbf{y})}{\alpha} \biggr)\biggr]
\end{equation}
converges to the real value of the mutual information $I(X;Y)$.
\end{lemma}

\begin{proof}
The solution to \eqref{eq:Lemma3_problem} is given by \eqref{eq:optimal_discriminator_2} of Lemma \ref{lemma:Lemma2}. Let $\delta^{(n)}=D^*-D^{(n)}$ be the displacement between the optimum discriminator $D^*$ and the obtained one $D^{(n)}$ at the iteration $n$, then
\begin{align}
& \hat{I}_{n,dDIME}(X;Y) = \; \mathbb{E}_{(\mathbf{x},\mathbf{y}) \sim p_{XY}(\mathbf{x},\mathbf{y})}\biggl[\log \biggl(\frac{D^{(n)}(\mathbf{x},\mathbf{y})}{\alpha} \biggr)\biggr] \nonumber \\ 
& = \; \mathbb{E}_{(\mathbf{x},\mathbf{y}) \sim p_{XY}(\mathbf{x},\mathbf{y})}\biggl[\log \biggl(\frac{D^*(\mathbf{x},\mathbf{y})-\delta^{(n)}}{\alpha} \biggr)\biggr] \nonumber \\ 
& = \; \mathbb{E}_{(\mathbf{x},\mathbf{y}) \sim p_{XY}(\mathbf{x},\mathbf{y})}\biggl[\log \biggl(\biggl(\frac{D^*(\mathbf{x},\mathbf{y})}{\alpha}\biggl)\biggl(1-\frac{\delta^{(n)}}{D^*(\mathbf{x},\mathbf{y})} \biggr)\biggr) \biggr] \nonumber \\ 
& = \; I(X;Y) + \mathbb{E}_{(\mathbf{x},\mathbf{y}) \sim p_{XY}(\mathbf{x},\mathbf{y})}\biggl[\log \biggl(1-\frac{\delta^{(n)}}{D^*(\mathbf{x},\mathbf{y})} \biggr) \biggr].
\end{align}
If the gradient descent method converges towards the optimum solution $D^*$, $\delta^{(n)} \rightarrow 0$ and therefore 
\begin{align}
& \hat{I}_{n,dDIME}(X;Y) \simeq \nonumber \\
& \simeq \; I(X;Y) - \delta^{(n)} \cdot \mathbb{E}_{(\mathbf{x},\mathbf{y}) \sim p_{XY}(\mathbf{x},\mathbf{y})}\biggl[\frac{1}{D^*(\mathbf{x},\mathbf{y})} \biggr] \nonumber \\ 
& = \; I(X;Y) - \delta^{(n)} \cdot \mathbb{E}_{(\mathbf{x},\mathbf{y}) \sim p_{XY}(\mathbf{x},\mathbf{y})}\biggl[\frac{p_{X}(\mathbf{x})\cdot p_Y(\mathbf{y})}{\alpha \cdot p_{XY}(\mathbf{x},\mathbf{y})} \biggr] \nonumber \\
& = \; I(X;Y) - \frac{\delta^{(n)}}{\alpha},
\end{align}
which implies in the asymptotic limit ($n\rightarrow +\infty$) that 
\begin{equation}
|I(X;Y)-\hat{I}_{n,dDIME}(X;Y)|\rightarrow 0.
\end{equation}
\end{proof}
\qedhere  

\begin{theorem}
Let $X\sim p_X(\mathbf{x})$ and $Y\sim p_{Y}(\mathbf{y}|\mathbf{x})$ be the channel input and output, respectively, such that $X = G(Z)$ with $Z\sim p_Z(\mathbf{z})$. Let $Y = H(X)$ with $H(\cdot)$ being the stochastic channel model and let $\pi(\cdot)$ be a permutation function such that $p_{\pi(Y)}(\pi(\mathbf{y})|\mathbf{x}) = p_{Y}(\mathbf{y})$. 
If $\mathcal{J}_{\alpha}(G,D)$, $\alpha>0$, is a value function defined as 
\begin{align}
\mathcal{J}_{\alpha}(G,D) =  \; & \alpha \cdot \mathbb{E}_{\mathbf{z} \sim p_{Z}(\mathbf{z})}\biggl[\log \biggl(D\biggl(G(\mathbf{z}),H(G(\mathbf{z}))\biggr)\biggr)\biggr] \nonumber \\
&  +\mathbb{E}_{\mathbf{z} \sim p_{Z}(\mathbf{z})}\biggl[-D\biggl(G(\mathbf{z}),\pi(H(G(\mathbf{z})))\biggr)\biggr],
\end{align}
then the channel capacity $C$ is the solution of
\begin{equation}
C = \max_{G} \max_{D} \frac{\mathcal{J}_{\alpha}(G,D)}{\alpha} + 1- \log(\alpha)
\end{equation}
\end{theorem}

\begin{proof}
From the hypothesis, the value function can be rewritten as
\begin{align}
\mathcal{J}_{\alpha}(G,D) = \; & \alpha \cdot \mathbb{E}_{(\mathbf{x},\mathbf{y}) \sim p_{XY}(\mathbf{x},\mathbf{y})}\biggl[\log \biggl(D(\mathbf{x},\mathbf{y})\biggr)\biggr] \nonumber \\
& +\mathbb{E}_{(\mathbf{x},\mathbf{y}) \sim p_{X}(\mathbf{x})\cdot p_Y(\mathbf{y})}\biggl[-D(\mathbf{x},\mathbf{y})\biggr].
\label{eq:value_function_p}
\end{align}
Given a fixed generator $G$, Lemma \ref{lemma:Lemma2} asserts that $\mathcal{J}_{\alpha}(G,D)$ is maximized for 
\begin{equation}
D(\mathbf{x},\mathbf{y})=D^*(\mathbf{x},\mathbf{y})=\alpha \cdot \frac{p_{XY}(\mathbf{x},\mathbf{y})}{p_{X}(\mathbf{x})\cdot p_Y(\mathbf{y})}.
\end{equation}
Substituting $D^*(\mathbf{x},\mathbf{y})$ in \eqref{eq:value_function_p} gives
\begin{align}
\mathcal{J}_{\alpha}(G,D^*) =  \; \alpha & \cdot \mathbb{E}_{(\mathbf{x},\mathbf{y}) \sim p_{XY}(\mathbf{x},\mathbf{y})}\biggl[\log \biggl(\alpha \cdot \frac{p_{XY}(\mathbf{x},\mathbf{y})}{p_{X}(\mathbf{x})\cdot p_Y(\mathbf{y})} \biggr)\biggr] \nonumber \\
& +\mathbb{E}_{(\mathbf{x},\mathbf{y}) \sim p_{X}(\mathbf{x})\cdot p_Y(\mathbf{y})}\biggl[-\alpha \cdot \frac{p_{XY}(\mathbf{x},\mathbf{y})}{p_{X}(\mathbf{x})\cdot p_Y(\mathbf{y})}\biggr],
\end{align}
and using the Lebesgue integral to compute the expectation
\begin{align}
\mathcal{J}_{\alpha}(G,D^*) = \; & \alpha \int_{\mathbf{y}} \int_{\mathbf{x}}{p_{XY}(\mathbf{x},\mathbf{y}) \log \biggl(\alpha  \frac{p_{XY}(\mathbf{x},\mathbf{y})}{p_{X}(\mathbf{x})\cdot p_Y(\mathbf{y})} \biggr) \diff \mathbf{x} \diff \mathbf{y}} \nonumber \\
& + \int_{\mathbf{y}} \int_{\mathbf{x}}{p_{X}(\mathbf{x}) p_Y(\mathbf{y}) \biggl(-\alpha \frac{p_{XY}(\mathbf{x},\mathbf{y})}{p_{X}(\mathbf{x})\cdot p_Y(\mathbf{y})}\biggr) \diff \mathbf{x} \diff \mathbf{y}},
\end{align}
where it is easy to recognize that the second term of the right hand side is equal to $-\alpha$, resulting in
\begin{align}
\mathcal{J}_{\alpha}(G,D^*) =  \; & \alpha \log(\alpha) - \alpha + \nonumber \\
& +\alpha \cdot \mathbb{E}_{(\mathbf{x},\mathbf{y}) \sim p_{XY}(\mathbf{x},\mathbf{y})}\biggl[\log \biggl(\frac{p_{XY}(\mathbf{x},\mathbf{y})}{p_{X}(\mathbf{x})\cdot p_Y(\mathbf{y})} \biggr)\biggr].
\end{align}
Finally, by definition of mutual information ${I}(X;Y)$, a maximization over the generator $G$ results in a mutual information maximization
\begin{equation}
\max_{G} \mathcal{J}_{\alpha}(G,D^*) = \max_{G} \bigl(\alpha \log(\alpha) - \alpha + \alpha \cdot {I}(X;Y)\bigr),
\end{equation}
hence, since $G$ models the distribution $p_X(\mathbf{x})$ and $\alpha$ is a constant,
\begin{equation}
\max_{G} \mathcal{J}_{\alpha}(G,D^*)  + \alpha - \alpha \log(\alpha) = \alpha \cdot \max_{p_X(\mathbf{x})} {I}(X;Y).
\end{equation}
\qedhere  
\end{proof}

\begin{lemma}
\label{lemma:Lemma5}
Let $f(u)=-\alpha \cdot \log(u)$, $\alpha>0$, be a convex, lower-semicontinuous function. Then, $f$ admits a convex conjugate function $f^*$, defined as
\begin{equation}
f^*(t) = \sup_{u\in \mathbb{R_+}} \{ut - f(u)\}
\end{equation}
with $t<0$ and expression
\begin{equation}
f^*(t) = -\alpha - \alpha \cdot \log \biggl(-\frac{t}{\alpha}\biggr).
\label{eq:fenchel_conj}
\end{equation}
\end{lemma}

\begin{proof}
From the Fenchel conjugate definition,
\begin{equation}
f^*(t) = \sup_{u\in \mathbb{R_+}} \{ut + \alpha \log(u)\},
\end{equation}
partial differentiation w.r.t. $u$ gives one stationary point in $u^*=-\alpha/t$ for negatives $t$. In particular, $u^*$ is a global maximum since 
\begin{equation}
\lim_{u\to 0^+} ut + \alpha \log(u)\ = \lim_{u\to +\infty} ut + \alpha \log(u)\ = -\infty.
\end{equation}
Hence, for $t\in \mathbb{R_{-}}$, 
\begin{equation}
f^*(t) = -\alpha +\alpha \cdot \log  \biggl(-\frac{\alpha}{t}\biggr).
\end{equation}
For $t\geq 0$, it is straightforward to verify that $f^*(t) = +\infty$.
\qedhere  
\end{proof}

\begin{theorem}
For any class of functions $\mathcal{D}$ mapping from the sample domain $\mathcal{X}$ to $\mathbb{R_+}$, the following lower bound on the mutual information holds with $\alpha>0$
\begin{align}
I(X;Y) \geq & \sup_{D\in \mathcal{D_+}} \biggl\{\mathbb{E}_{(\mathbf{x},\mathbf{y}) \sim p_{XY}(\mathbf{x},\mathbf{y})}\bigl[\log (D(\mathbf{x},\mathbf{y}))\bigr] \nonumber \\
& - \frac{1}{\alpha}\cdot \mathbb{E}_{(\mathbf{x},\mathbf{y}) \sim p_{X}(\mathbf{x})\cdot p_{Y}(\mathbf{y})}\bigl[D(\mathbf{x},\mathbf{y})\bigr] +1 -\log(\alpha)\biggr\}
\end{align}
\end{theorem}

\begin{proof}
In the KL divergence expression, dividing and multiplying by $\alpha$ yields
\begin{align}
D_{KL}(p||q) & = \int_{\mathbf{x}} {p_{X}(\mathbf{x}) \log \biggl(\frac{p_{X}(\mathbf{x})}{q_{X}(\mathbf{x})} \biggr) \diff \mathbf{x}} \nonumber \\
& = \int_{\mathbf{x}} {\frac{p_{X}(\mathbf{x})}{\alpha} \biggl(-\alpha \log \biggl(\frac{q_{X}(\mathbf{x})}{p_{X}(\mathbf{x})} \biggr)\biggr) \diff \mathbf{x}}
\end{align}
so that it is possible t rewrite the logarithmic term of the KL divergence in terms of its conjugate, in particular 
\begin{align}
& D_{KL}(p||q) = \int_{\mathbf{x}} {\frac{p_{X}(\mathbf{x})}{\alpha} \sup_{\mathbf{y}} \biggl\{ \mathbf{y} \cdot \frac{q_{X}(\mathbf{x})}{p_{X}(\mathbf{x})} - f^*(\mathbf{y}) \biggr\} \diff \mathbf{x}} \nonumber \\
& = \sup_{\mathbf{y}} \biggl\{ \int_{\mathbf{x}} {\frac{q_{X}(\mathbf{x})}{\alpha}\mathbf{y}\diff \mathbf{x}} - \int_{\mathbf{x}}{\frac{p_{X}(\mathbf{x})}{\alpha} f^*(\mathbf{y}) \diff \mathbf{x}} \biggr\} \nonumber \\
& \geq  \sup_{T\in \mathcal{D_{-}}} \biggl\{ \frac{1}{\alpha} \cdot \mathbb{E}_{\mathbf{x} \sim q_{X}(\mathbf{x})}\bigl[T(\mathbf{x})\bigr]-\frac{1}{\alpha}\cdot  \mathbb{E}_{\mathbf{x} \sim p_{X}(\mathbf{x})}\bigl[f^*\bigl(T(\mathbf{x})\bigr)\bigr]\biggr\} \nonumber \\
& = \sup_{D\in \mathcal{D_{+}}} \biggl\{-\frac{1}{\alpha}\cdot  \mathbb{E}_{\mathbf{x} \sim p_{X}(\mathbf{x})}\bigl[f^*\bigl(-D(\mathbf{x})\bigr)\bigr] - \frac{1}{\alpha} \cdot \mathbb{E}_{\mathbf{x} \sim q_{X}(\mathbf{x})}\bigl[D(\mathbf{x})\bigr] \biggr\},
\end{align}
where the inequality is introduced to take into account the fact that $\mathcal{D_+}$ is a subset of all possible functions, and $D = -T$. From Lemma \ref{lemma:Lemma5}, the Fenchel conjugate of $f(u)=-\alpha \log(u)$, $\alpha>0$, takes the form as in \eqref{eq:fenchel_conj}. Substituting it above yields
\begin{align}
D_{KL}(p||q) \geq & \sup_{D\in \mathcal{D_{+}}} \biggl\{-\frac{1}{\alpha}\cdot  \mathbb{E}_{\mathbf{x} \sim p_{X}(\mathbf{x})}\biggl[-\alpha - \alpha \cdot \log\biggl(\frac{D(\mathbf{x})}{\alpha}\biggr)\biggr] \nonumber \\
&  - \frac{1}{\alpha} \cdot \mathbb{E}_{\mathbf{x} \sim q_{X}(\mathbf{x})}\bigl[D(\mathbf{x})\bigr] \biggr\} \nonumber \\
= & \sup_{D\in \mathcal{D_+}} \biggl\{ 1 -\log(\alpha)+ \mathbb{E}_{\mathbf{x} \sim p_{X}(\mathbf{x})}\bigl[\log (D(\mathbf{x}))\bigr] \nonumber \\
&  - \frac{1}{\alpha}\cdot \mathbb{E}_{\mathbf{x} \sim q_{X}(\mathbf{x})}\bigl[D(\mathbf{x})\bigr] \biggr\}.
\end{align}

The alpha variational representation of the KL divergence can be extended to lower bound the mutual information. Indeed, given that $I(X;Y)=D_{KL}(p_{XY}(\mathbf{x},\mathbf{y})||p_{X}(\mathbf{x})\cdot p_Y(\mathbf{y}))$, the thesis follows

\begin{align}
I(X;Y) \geq & \sup_{D\in \mathcal{D_+}} \biggl\{1 -\log(\alpha) + \mathbb{E}_{(\mathbf{x},\mathbf{y}) \sim p_{XY}(\mathbf{x},\mathbf{y})}\bigl[\log (D(\mathbf{x},\mathbf{y}))\bigr] \nonumber \\ 
& - \frac{1}{\alpha}\cdot \mathbb{E}_{(\mathbf{x},\mathbf{y}) \sim p_{X}(\mathbf{x})\cdot p_{Y}(\mathbf{y})}\bigl[D(\mathbf{x},\mathbf{y})\bigr] \biggr\}.
\end{align}

\qedhere  
\end{proof}

\subsection{Experimental Details}
\label{sec:experiment_details}

In this section, we describe the implementation details of DIME and CORTICAL.
\subsubsection{DIME implementation details}
For all the mutual information estimators proposed, we consider a joint architecture that concatenates the channel input-output, $\mathbf{x},\mathbf{y}$. The permutation function $\pi(\cdot)$ is implemented as follows: given one realization sample $(\mathbf{x},\mathbf{y})$ of the joint distribution $p_{XY}(\mathbf{x},\mathbf{y})$, we randomly permute the $d$ elements of $\mathbf{y}$ such that the element-wise pairs are $(x_i,y_j)$, with $i\neq j$. In this way, $(\mathbf{x},\pi(\mathbf{y}))$ represents a realization of the product of the marginal distributions $p_{X}(\mathbf{x})\cdot p_{Y}(\mathbf{y})$.

The discriminative estimators possess the same architecture, a two layer multilayer perceptron (MLP) neural network with $100$ neurons in each layer with the ReLU activation function. The only difference resides in the final layer where we use a linear dense layer for $I_{MINE}$, while a sigmoid layer for $I_{iDIME}$ and a softplus layer for $I_{dDIME}$, defined as
\begin{equation*}
 D^*(\mathbf{x},\mathbf{y}; \theta) = \log(1+\exp(\mathbf{W}^*_L \cdot \mathbf{z}_{L-1}+b^*_L)).
\end{equation*}
Each discriminator used in the mutual information estimator has been trained for $n = 5$k iterations with a batch size of $512$, for every SNR value. We use Adam as the optimizer with a learning rate of $0.002$.  To avoid training bias, we repeat each training $10$ times.
Tab.\ref{tab:dime_architecture} provides all the architecture and parameters details.

In the testing phase, we compute a total of $10000$ instantaneous estimations of the mutual information (single batch as input) for all the estimators and display by means of the shadow regions of Fig.\ref{fig:i_estimator} and \ref{fig:j_estimator} the standard deviation. This shows the variability of the instantaneous estimates which are then averaged to provide the final mutual information estimate.
% We report the results for two values of the SNR in Tab.\ref{tab:results}.

\begin{table}[h]
	\centering
	\caption{Discriminative networks architecture and training parameters.}
	\begin{tabular}{ p{4cm}|p{2cm}|p{2cm}} 
		\toprule
		\textbf{Layer} & \textbf{Output dimension } 		& \textbf{Activation function} \\
		\midrule
		\textbf{i-DIME} & &\\
		Input $[\mathbf{x},\mathbf{y}]$ & $2\cdot d$ & \\ 
		Fully connected & 100 & ReLU \\ 
		Dropout &0.3&  \\ 
		Fully connected &100& ReLU  \\ 
		Fully connected & 1 & Sigmoid  \\  		\midrule

		\textbf{d-DIME} & &\\
		Input $[\mathbf{x},\mathbf{y}]$ & $2\cdot d$ & \\ 
		Fully connected & 100 & ReLU \\ 
		Dropout &0.3&  \\ 
		Fully connected &100& ReLU  \\ 
		Fully connected & 1 & Softplus  \\  		\midrule
		\textbf{MINE} & &\\
		Input $[\mathbf{x},\mathbf{y}]$ & $2\cdot d$ & \\ 
		Fully connected & 100 & ReLU \\ 
		Dropout &0.3&  \\ 
		Fully connected &100& ReLU  \\ 
		Fully connected & 1 &   \\  		\midrule
		Batch size &  \multicolumn{2}{c}{512}  \\ 
		Number of training iterations &  \multicolumn{2}{c}{5000}  \\ 
		Learning rate &  \multicolumn{2}{c}{0.002}   \\ 
		Optimizer &  \multicolumn{2}{c}{Adam ($\beta_1$ = 0.5, $\beta_2$ = 0.999)}  \\ 		\midrule

		%\bottomrule	
	\end{tabular}
	
	\label{tab:dime_architecture}
\end{table}

\begin{comment}
\renewcommand{\arraystretch}{1.2}
\begin{table*}

	\centering
	\caption{Mutual information estimation between $2$-d and $10$-d Gaussians at $\text{SNR}=-5$ dB and $\text{SNR}=10$ dB.}
\begin{tabular}{|l|c|c|c|c|}
\cline{1-5}
Mutual Information Estimator & \multicolumn{2}{c|}{$2$-d Gaussians} & \multicolumn{2}{c|}{$10$-d Gaussians}  \\ \cline{1-5}
          & $-5$ dB            & $10$ dB            & $-5$ dB            & $10$ dB              \\ \cline{1-5}
True Mutual Information $I(X;Y)$  & $0.396$  & $3.459$   & $1.982$  & $17.30$   \\ 
i-DIME ($\hat{I}_{iDIME}$)   & $0.31 \pm 0.63$   & $2.91 \pm 1.22$  & $1.39 \pm 1.39$  & $10.5 \pm 2.79$    \\ 
d-DIME ($\hat{I}_{dDIME},\alpha=0.1$)  & $0.35 \pm 0.66$    & $3.32 \pm 1.18$  & $1.66 \pm 1.34$  & $7.64 \pm 1.28$    \\
d-DIME ($\hat{I}_{dDIME},\alpha=1$)  & $0.38 \pm 0.67$   & $3.21 \pm 1.14$  & --  & --    \\
d-DIME ($\hat{I}_{dDIME}, \alpha=10$)  & $0.36 \pm 0.66$   & $3.25 \pm 1.10$  & -- & --    \\
d-DIME ($\tilde{I}_{dDIME},\alpha=0.1$)  & $0.37 \pm 1.03$   & $3.39 \pm 3.72$  & $1.63 \pm 2.57$ & $8.62 \pm 6.57$    \\ 
d-DIME ($\tilde{I}_{dDIME},\alpha=1 = \hat{I}_{NWJ}$)  & $0.38 \pm 1.02$   & $3.35 \pm 3.60$  & --  & --    \\ 
d-DIME ($\tilde{I}_{dDIME},\alpha=10$)  & $0.39 \pm 1.00$   & $3.41 \pm 3.37$  & --  & --   \\ 
MINE ($\hat{I}_{MINE}$)   & $0.37 \pm 0.18$   & $3.06 \pm 0.91$  & $1.50 \pm 0.50$  & $8.38 \pm 5.32$    \\ \cline{1-5}

\end{tabular}
	\label{tab:results}

\end{table*}
\end{comment}
\subsubsection{CORTICAL implementation details}

In the CORTICAL framework, we exploit the architectures discussed in Tab.\ref{tab:dime_architecture} for the discriminator implementation. The encoder/generator, instead, consists of a three layer MLP with $100$ neurons in each with ReLU activation. The final layer is a linear dense layer to allow synthesizing any type of distribution followed by a batch normalization layer to account for the power constraint. We alternate a generator training iteration every $10$ training iterations of the discriminator, for a total number of $500$ generator training iterations. In the continuous coding case, the input vector $\mathbf{z}$ is sampled from a $30$-dimensional Gaussian distribution. For the discrete case, the input vector $\mathbf{z}$ is sampled from a multivariate Bernoulli $3$-dimensional distribution ($M=8$) with probability $p=0.5$. In other words, the input vector $\mathbf{z}$ is a sequence of bits that corresponds to the representation in base $2$ of $s$, where $s \in \mathbb{N}$ and $s = \{0,1,\dots,M-1\}$. It is interesting to notice that the dimension of $\mathbf{z}$ is equal to $\log_2(M)$ and not to $M$, which is the common case when dealing with classification problems solved via the cross-entropy cost function and one-hot coding for the softmax last layer (as it is done in the autoencoder for communications scenario).

Tab.\ref{tab:cortical_architecture} provides all the architecture and parameters details of the used generator in the CORTICAL framework.

\begin{table}
	\centering
	\caption{CORTICAL generator architecture and training parameters.}
	\begin{tabular}{ p{4cm}|p{2cm}|p{2cm}} 
		\toprule
		\textbf{Layer} & \textbf{Output dimension } 		& \textbf{Activation function} \\
		\midrule
		\textbf{CORTICAL generator} & &\\
		Input $\mathbf{z}$ & 30 (continuous) / 3 (discrete) & \\ 
		Fully connected & 100 & ReLU \\ 
		Fully connected &100& ReLU  \\ 
		Fully connected &100& ReLU  \\ 
		Fully connected & $d$ &   \\  		\midrule

		Batch size &  \multicolumn{2}{c}{512}  \\ 
		Number of training iterations &  \multicolumn{2}{c}{500}  \\ 
		Learning rate &  \multicolumn{2}{c}{0.0002}   \\ 
		Optimizer &  \multicolumn{2}{c}{Adam ($\beta_1$ = 0.5, $\beta_2$ = 0.999)}  \\ 		\midrule

		%\bottomrule	
	\end{tabular}
	
	\label{tab:cortical_architecture}
\end{table}

\end{document}